\documentclass{amsart}
\usepackage[utf8]{inputenc}
\usepackage[T1]{fontenc}
\usepackage{xspace}
\usepackage{graphicx}
\usepackage{tikz}
\usepackage{thmtools}
\usepackage{hyperref}

\usepackage{floatrow}
\usepackage{enumerate}
\usepackage{mathtools}

\usepackage{algorithm}
\usepackage{algpseudocode}

\usepackage{amsthm}
\usepackage{enumitem}
\newcommand{\BR}[1]{BR$_ #1$}
\newcommand{\KR}[1]{KR$_ #1$}

 \hypersetup{
 pdftitle = {Pushing the frontiers of subexponential FPT time for Feedback Vertex Set},
 pdfauthor=  {Gaétan Berthe, Marin Bougeret, Daniel Gonçalves, Jean-Florent Raymond},
  colorlinks = true,
 linkcolor = black!40!red,
 citecolor = black!40!green
 }

\usepackage[a4paper, margin=2.5cm]{geometry}
\usepackage{amsmath,amssymb}
\usepackage{thm-restate}
\usepackage[textsize=small]{todonotes}
\usepackage{comment}
\usepackage{euscript}

\declaretheorem[name=Theorem, numberwithin=section]{theorem} 
\declaretheorem[name=Lemma, sibling=theorem]{lemma}

\declaretheorem[name=Definition, sibling=theorem]{definition}
\declaretheorem[name=Corollary, sibling=theorem]{corollary}
\declaretheorem[name=Claim, sibling=theorem]{claim}
\declaretheorem[name=Remark, style=remark, sibling=theorem]{remark}
\declaretheorem[name=Question, style=remark, sibling=theorem]{question}

\newtheorem*{main-thm}{Main Theorem}

\newcommand{\algopart }{\texttt{Partitionning-Algorithm}~ }

\newcommand{\RR}{\mathbb{R}}

\newcommand{\G}{\mathcal{G}}
\newcommand{\D}{\mathcal{D}} 
 
\newcommand{\T}{\mathcal{T}} 
\newcommand{\I}{\mathcal{I}}

\newcommand{\mC}{\mathcal{MC}} 
\newcommand{\eps}{\varepsilon}

\newcommand{\mS}{\mathcal{S}}

\newcommand{\mR}{{\mathcal{R}}}

\newcommand{\dM}{d_M}
\newcommand{\pom}{\partial_{\overline{M}}}
\newcommand{\bom}{b_{\overline{M}}}
\newcommand{\mT}{\mathcal{T}}
\newcommand{\mH}{\mathcal{H}}
\newcommand{\KttTilde}{\widetilde{K_{t,t}}}

\renewcommand{\d}{d}
\newcommand{\f}{\tau}
\newcommand{\tdp}{T_{DP}}
\newcommand{\CDP}{C_{DP}}

\newcommand{\intv}[2]{\left \{ #1, \dots, #2 \right \}}

\DeclareMathOperator\grid{\boxplus}

\DeclareMathOperator\tw{\textsf{tw}}
\DeclareMathOperator\poly{\textsf{poly}}
\DeclareMathOperator\db{\mu}
\DeclareMathOperator\TO{\Tilde{\O}}

\newcommand{\fvs}{\textsc{Feedback Vertex Set}\xspace}
\newcommand{\afvs}{\textsc{$(r,\G)$-Annotated Feedback Vertex Set}\xspace}

\newcommand{\FVS}{\textsc{FVS}\xspace}

\newcommand{\AFVS}{\textsc{$(r,\G)$-Ann-FVS}\xspace}

\newcommand{\ruleref}[1]{\hyperref[#1]{\ref*{#1}}}

\renewcommand{\O}{\mathcal{O}}

\def\cqedsymbol{\ifmmode$\lrcorner$\else{\unskip\nobreak\hfil
\penalty50\hskip1em\null\nobreak\hfil$\lrcorner$
\parfillskip=0pt\finalhyphendemerits=0\endgraf}\fi} 
\newcommand{\cqed}{\renewcommand{\qed}{\cqedsymbol}}

\usepackage{marvosym}
\newcommand{\apx}[1]{\hyperref[#1]{\LeftScissors}}

\graphicspath{{fig/}}

\title{
Pushing the frontiers of subexponential FPT time for Feedback Vertex Set}
\author[G.~Berthe]{Gaétan Berthe}
\author[M.~Bougeret]{Marin Bougeret}
\author[D.~Gonçalves]{Daniel Gonçalves}
\address[G.~Berthe, M.~Bougeret, D.~Gonçalves]{LIRMM, Université de Montpellier, CNRS, Montpellier, France.}
\email{firstname.lastname@lirmm.fr}

\author[J.-F.~Raymond]{Jean-Florent Raymond}
\address[J.-F.~Raymond]{CNRS, ENS de Lyon, Université Claude Bernard Lyon 1, LIP, UMR5668,
  Lyon, France.}
\email{jean-florent.raymond@cnrs.fr}

\date{April 2025}

\begin{document}

\maketitle

\begin{abstract}
The paper deals with the \fvs problem parameterized by the solution size. Given a graph $G$ and a parameter $k$, one has to decide if there is a set $S$ of at most $k$ vertices such that $G-S$ is acyclic.
Assuming the Exponential Time Hypothesis, it is known that \FVS cannot be solved in time $2^{o(k)}n^{\O(1)}$ in general graphs. To overcome this, many recent results considered \FVS restricted to particular intersection graph classes and provided such $2^{o(k)}n^{\O(1)}$ algorithms.

In this paper we provide generic conditions on a graph class for the existence of an algorithm solving \FVS in subexponential FPT time, i.e. time $2^{k^\eps} \mathop{\rm poly}(n)$, for some $\eps<1$, where $n$ denotes the number of vertices of the instance and $k$ the parameter.
On the one hand this result unifies algorithms that have been proposed over the years for several graph classes such as planar graphs, map graphs, unit-disk graphs, pseudo-disk graphs, and string graphs of bounded degree. On the other hand, it extends the tractability horizon of \FVS to new classes that are not amenable to previously used techniques, in particular intersection graphs of ``thin'' objects like segment graphs or more generally $s$-string graphs.
\end{abstract}

\section{Introduction}
\label{sec:intro}

\subsection{Context}

Given an $n$-vertex graph $G$ and a parameter $k\in \mathbb{N}$, the \fvs problem (\FVS for short) asks whether there exists a set $S$ of at most $k$ vertices such that $G-S$ has no cycle. 
This is a fundamental decision problem in graph theory and one of Karp's 21 NP-complete problems. Because of its hardness in a classical setting, the problem has been widely studied within the realm of parameterized complexity.
This line of research aims to investigate the existence of \emph{FPT algorithms}, i.e., algorithms that run in time $f(k)\cdot n^{\O(1)}$, for some computable function $f$. Such algorithms provide a fine-grained understanding on the time complexity of a problem and describe regions of the input space where the problem can be solved in polynomial time. Note that it is crucial here to obtain good bounds on the function $f$ since the (potentially) super-polynomial part of the running time is confined in the $f(k)$ term. In this direction it was proved that under the Exponential Time Hypothesis of Impagliazzo, Paturi and Zane, FVS does not admit an algorithm with running time $2^{o(k)}n^{\O(1)}$ (see \cite{Cygan2015Book}). Nevertheless certain classes of graphs (typically planar graphs) have been shown to admit algorithms with running times of the form $2^{\O(k^{\eps})} n^{\O(1)}$ (for some $\eps<1$), i.e., where the contribution of the parameter $k$ is subexponential. Such algorithms are called \emph{subexponential parameterized algorithms} and they are the topic of this paper.
Among the numerous existing results on this theme, there are two main directions of research: to improve the running times in the classes where an algorithm is already known, or to extend the tractability horizon of \FVS by providing more general settings where subexponential FPT algorithms exist.

We are here interested by the second direction, which can be summarized by the following question.
\begin{question}\label{q:main}
What are the most general graph classes where \FVS admits a subexponential parameterized algorithm ?
\end{question}

Historically, a primary source of graph classes studied to make progress on the above question was geometric intersection graphs.
In an \emph{intersection graph}, each vertex corresponds to a subset of some ambient space, and two vertices are adjacent if and only if the subsets intersect. Taking the Euclidean plane as the ambient space, many graph classes can be defined by setting restrictions on the subsets used to represent the vertices. One can for instance consider intersection graphs of disks in the plane, or segments, or Jordan arcs.\footnote{In the following, those are called strings.}
With such subsets, one defines the class of {\em disk graphs}, {\em segment graphs}, and {\em string graphs} respectively. It is also often the case that there are conditions dealing with all the $n$ subsets representing the vertices of a given graph. For example, if we consider disks (resp. segments) one can ask those to have the same diameter (resp. to use at most $d$ different slopes), and this defines the class of {\em unit disk graphs} (resp. {\em $d$-DIR graphs}). When considering strings, one possible property is that any string has at most $d$ points shared with the other considered strings. This defines string graphs with {\em edge-degree at most~$d$}. It was recently proved that for a string graph, having bounded edge-degree is in fact equivalent to having bounded degree~\cite{karol2025}.
A weaker condition is to ask every pair of the considered strings to intersect on at most $s$ points, this defines {\em $s$-string graphs}.
This class generalizes several natural classes such as planar graphs, map graphs, unit-disk graphs, segment graphs, string graphs of bounded degree, and intersection graphs of $\alpha$-convex bodies that exclude a fixed subgraph (see \cite{matouvsek2014string, kratochvil1994intersection, baste2022contraction}).

For all these graph classes, \FVS is NP-complete, and actually under ETH none of them admits a $2^{o(\sqrt{n})}$-time algorithm. Indeed, this lower bound was given in~\cite{deberg_unit} for induced grid graphs, which form a subclass of unit disk graphs and 2-DIR graphs. On the other hand, for each of the aforementioned classes there is an algorithm solving \FVS in subexponential time. More precisely, this algorithm applies to any string graph and runs in time $2^{\TO(n^{2/3})}$~\cite{bonnet2019optimality}.\footnote{The notation $\TO$ ignores polylogarithmic factors, i.e. we write $g(x) = \TO(f(x))$ if for some $c$ we have $g(x) = \O(f(x) \cdot \log^c x)$.}

Regarding subexponential parameterized algorithms, the case of unit disk graphs was settled with an algorithm whose running time matches the ETH lower bound~\cite{an_unit_21}. This result uses the fact that these graphs admit vertex partitions into cliques such that each of these cliques is adjacent to only a constant number of the other cliques. Such a property does not hold for the other graph classes mentioned above. 
However, other techniques have been developed to deal with the other aforementioned classes such as the classes of bounded edge-degree string graphs~\cite{baste2022contraction}, contact-segment and square graphs~\cite{berthe24ASQGM}, disk graphs~\cite{lokSODA22,Faster2023Shinwoo}, or the pseudo-disk graphs~\cite{FVS-WG}.
Note that when dealing with classes of intersection graphs, the representation of the input (if known) could be used by the algorithm.
Some of these algorithms are {\em robust}, meaning that the input graph $G$ is provided using one of the classical graph data structures, where there is no indication of the intersection model of $G$. Because the recognition problem is difficult for most of the classes discussed above, robustness is a substantial advantage.

\subsection{Our contribution}
Toward answering \autoref{q:main}, we identify sufficient conditions for a graph class (then said to be \emph{nice}) to admit a subexponential parameterized algorithm for \FVS. 
As we will see later these conditions are satisfied by several natural graph classes, some of which were not known to admit a subexponential parameterized algorithm prior to this work.

Let us now provide some intuition behind the conditions we require for a \emph{nice} graph class. We discuss here the similarities between these conditions and classical studied properties, while the reasons why these conditions help to get a subexponential parameterized algorithm for \FVS are examined in \autoref{sec:techniques}.
A starting point is to review known results about the class of string graphs, which constitutes a good candidate to answer the previous question. In particular, the following results are known for string graphs. For a graph $H$, let us say that a graph is \emph{$H$-free} if it does not contain $H$ as a subgraph.
\begin{theorem}[\cite{lee2016separators}]\label{th:Lee}
    There exists a constant $c$ such that for $r>0$ it holds that every $K_{r,r}$-free string graph on $n$ vertices has at most $cr(\log r) n$ edges.
\end{theorem}
\begin{theorem}[\cite{lee2016separators,DVORAK2019137}]\label{th:tw}
$K_{r,r}$-free string graphs on $n$ vertices have treewidth $\O(\sqrt{nr\log r})$.
\end{theorem}

These results\footnote{Note that an error was found in the proof of the above results in~\cite{lee2016separators}, but a claim by the author was made that the proof can be corrected in the case of string graphs (see~\cite{Bonnet2024}). Moreover the earlier bound in~\cite{Matousek14} yields similar results, up to logarithmic factors.
} are interesting in our case, as a simple folklore branching allows us to reduce the problem to the case where the instance $(G,k)$ of \FVS is $K_{r,r}$-free for $r=\lceil k^\eps \rceil$.
Thus, among the conditions required for a graph class to be nice,
two of them correspond to a relaxed version of the above theorems.


Our last main condition is related to neighborhood complexity.
A graph class $\mathcal{G}$ has \emph{linear neighborhood complexity (with ratio $c$)} if for any graph $G\in \mathcal{G}$ and any $X \subseteq V(G)$, we have that $|\{N(v) \cap X, v \in V(G)\}| \le c|X|$. 
It is known that graph classes of bounded expansion have linear neighborhood complexity~\cite{reidl2019characterising} as well as graph classes of bounded twin-width~\cite{bonnet2024neighbourhood}.
In previous works on parameterized subexponential algorithms \cite{Lokshtanov23Approx,Faster2023Shinwoo,FVS-WG,berthe24ASQGM}, it appeared useful that the considered graphs have the property that,
if $G$ is $K_{r,r}$-free (or even $K_r$-free), then for any $X \subseteq V(G)$, we have that $|\{N(v) \cap X, v \in V(G)\}| \le r^{\O(1)}|X|$.
Notice that this is slightly stronger than requiring that a class that is $K_{r,r}$-free (or $K_r$-free) for a fixed $r$ has linear neighborhood complexity, as it is important for our purpose that the dependency in $r$ is polynomial. 
We point out that $K_{r,r}$-free string graphs have bounded-expansion (by~\autoref{th:tw} and~\cite{Esperet2018}), hence linear neighborhood complexity, however this does not imply that the dependency in $r$ is polynomial.
Thus, our last main condition (called \emph{bounded tree neighborhood complexity}) can be seen as a slightly stronger version of this ``polynomially dependent'' neighborhood complexity.
Let us now proceed to the formal definitions.

\begin{definition}\label{def:NCtrees}
We say that a graph class $\mathcal{G}$ has \emph{bounded tree neighborhood complexity} (with parameters  $\alpha,f_1,f_2$) if there exist an integer $\alpha$ and two polynomial functions $f_1,f_2$ such that the following conditions hold.
For every $r$, every $K_{r,r}$-free graph $G \in \mathcal{G}$, every set $A\subseteq V(G)$ and every family $\mT$ of disjoint non-adjacent\footnote{Two vertex subsets are \emph{non-adjacent} in a graph if there is no edge from one to the other, see \autoref{subsec:basic}.} vertex subsets of $G-A$, each inducing a tree:
\begin{enumerate}
    \item $|\{N_A(T),~T\in \mT \}|\leq f_1(r) |A|^{\alpha}$, where $N_A(T)$ denotes the neighbors of the vertices of $T$ in $A$, and%
    \label{def:NCtrees:item1}
    \item $|\{N_A(T),~T\in \mT\}|\leq f_2(r,p,m)|A|$, where $p$ and $m$ denote the maximum  over all $T\in \mT$ of $|N_A(T)|$ and $|T|$ respectively. \label{def:NCtrees:item2}
\end{enumerate}
\end{definition}

\begin{definition}\label{def:prop}
\renewcommand{\theenumi}{C\arabic{enumi}}
We say that an hereditary graph class $\mathcal{G}$ is \emph{nice} (for parameters $\alpha,f_1,f_2, \delta, \f, \d$) if all the following conditions hold: 
\begin{enumerate}
    \item $\mathcal{G}$ is stable by contraction of an edge between degree-two vertices that do not belong to a triangle.\label{def:prop:deg2}
    \item $\mathcal{G}$ has bounded tree neighborhood complexity (for some parameters $\alpha,f_1,f_2$).\label{def:prop:nc}
    \item There exist a constant $\delta < 1$ and a polynomial function $\f(r)$ such that for any $K_{r,r}$-free graph $G \in \mathcal{G}$, $\tw(G) \le \f(r)\cdot n^{\delta}$.\label{def:prop:tw}
    \item The density $|E(G)|/|V(G)|$, of any $K_{r,r}$-free graph $G \in \mathcal{G}$, is upper bounded by a polynomial function $\d(r)$. Without loss of generality we will assume $\d(r)\geq r$.\label{def:prop:sparse}
\end{enumerate}
\end{definition}

Our main result is the following. 
\begin{main-thm}\label{main-thm}
For every nice hereditary graph class $\mathcal{G}$ there is a constant $\eta<1$ such that FVS can be solved in $\mathcal{G}$ in time $2^{\O(k^\eta)}\cdot  n^{\O(1)}$.
\end{main-thm}

Actually, we describe a single algorithm parameterized by the parameters of the nice classes it applies to. These parameters are used to make choices during the execution of the algorithm and they are used to define $\eta$.
The techniques used to prove the above result are discussed in \autoref{sec:techniques}. For the time being, let us focus on consequences.
As hinted above, being nice is a natural property shared by several well-studied classes of graphs. In particular we show that it is the case for $s$-string graphs and pseudo-disk graphs, hence we have the following applications.

\begin{restatable}{corollary}{maincor}\label{cor:main}
There exists $\eta<1$, such that for all $s$ there is a robust algorithm solving \FVS in time $2^{\TO \left (s^{\O(1)}k^{\eta}\right )}n^{\O(1)}$ for $n$-vertex $s$-string graphs.
\end{restatable}

\begin{restatable}{corollary}{maincortwo}\label{cor:main2}
There exists $\eta<1$, such that there is a robust algorithm solving \FVS in time $2^{\O(k^{\eta})}n^{\O(1)}$ for $n$-vertex pseudo-disk graphs.
\end{restatable}

Observe that the two corollaries above encompass a wide range of classes of geometric intersection graphs for which subexponential parameterized algorithms have been given in previous works such as planar graphs, map graphs, unit-disk graphs, disk graphs, or more generally pseudo-disk graphs, and string graphs of bounded degree. In this sense our main result unifies the previous algorithms.

Also, it captures new natural classes such as segment graphs, or more generally $s$-string graphs, where previous tools were not sufficient (as discussed in \autoref{sec:techniques}). We point out that before this work, the existence of subexponential parameterized algorithm for \FVS was open even for the very restricted class of $2$-DIR graphs.


\subsection{Basic notation}\label{subsec:basic}
In this paper logarithms are binary and all graphs are non-oriented and simple.
Unless otherwise specified we use standard graph theory terminology, as in \cite{diestel2005graph} for instance. For a graph $G$, and $v\in V(G)$, we denote $N_G(v)$ the neighbors of $v$.
We omit the subscript when it is clear from the context. For $A\subseteq V(G)$, we use the notation $N(A)=\left(\cup_{v\in A}N(v)\right)\setminus A$ and denote $G[A]$ the subgraph induced by $G$ on $A$. For $v\in V(G)$ and $B\subseteq V(G)$, we denote $N_B(v)=N(v)\cap B$ and for $A\subseteq V(G)$ we denote $N_B(A)=N(A)\cap B$, with the additional notation $d_B(A)=|N_B(A)|$. For a graph $H$ we say that $G$ is \emph{$H$-free} if $H$ is not a subgraph of $G$.
Two disjoint vertex subsets or subgraphs $Z,Z'$ of a graph $G$ are said to be \emph{non-adjacent} (in $G$) if there is no edge in $G$ with one endpoint in $Z$ and the other in $Z'$.

\subsection{Organisation of the paper}
In \autoref{sec:techniques} we explain why the approaches developed for other intersection graph classes in the papers \cite{Faster2023Shinwoo,baste2022contraction,FVS-WG,lokSODA22} do not apply here and present the main ideas behind our algorithm. 
 \autoref{sec:algo} and \autoref{sec:analysis} are devoted to the description and analysis of the algorithm. In \autoref{sec:applications} we provide applications of our main theorem by showing in particular that $s$-string graphs are nice.  Finally, in \autoref{sec:ccl} we discuss open problems and possible extensions of the approach developed here.

\section{Our techniques}
\label{sec:techniques}

\subsection{Why bidimensionality fails and differences with classes of ``fat'' objects}
Even if our goal is to abstract from a specific graph class, let us consider in this overview the class of $2$-DIR graphs, corresponding to the intersection graphs of vertical or horizontal segments in the plane. As these objects are non ``fat''\footnote{A regions $R$ of the plane is said to be \emph{$\alpha$-fat} if the radius of smallest disk enclosing $R$ is at most $\alpha$ times larger than the radius of the largest disk enclosed in $S$. A family of regions of the plane is then said to be \emph{fat} if there exists $\alpha$ such that all the elements of the family are $\alpha$-fat.} and can cross (unlike pseudo-disks), this class constitutes a good candidate to exemplify the difficulties.

A commonly used approach is as follows.
Given an instance $(G,k)$ of \FVS, we compute first in polynomial time a $2$-approximation, implying that we either detect a \textsc{No}-instance, or define a set $M$ with $|M|\le 2k$ and such that $G-M$ is a forest.
The goal is then to reduce the input, using kernelization or subexponential branching rules, to equivalent instances $(G',k')$ with small treewidth, that is $\tw(G')=\O(k^{1-\eps})$. As \FVS can be solved in $2^{\TO(\tw(G'))}n^{\O(1)}$ using a classical dynamic programming approach, we get a subexponential parameterized algorithm.
Thus, one has to find a way to destroy in $G$ the obstructions preventing a small treewidth.
A first type of obstructions is $K_r$ and $K_{r,r}$, which are easy to handle as there are folklore subexponential branchings when $r=k^{\eps}$.
Now, one can see the hard part is destroying a $K_{r,r}$
hidden (as a minor for example) (see \autoref{fig:introKrr}). 

\begin{figure}[!ht]
    \centering
    \includegraphics[width=\textwidth]{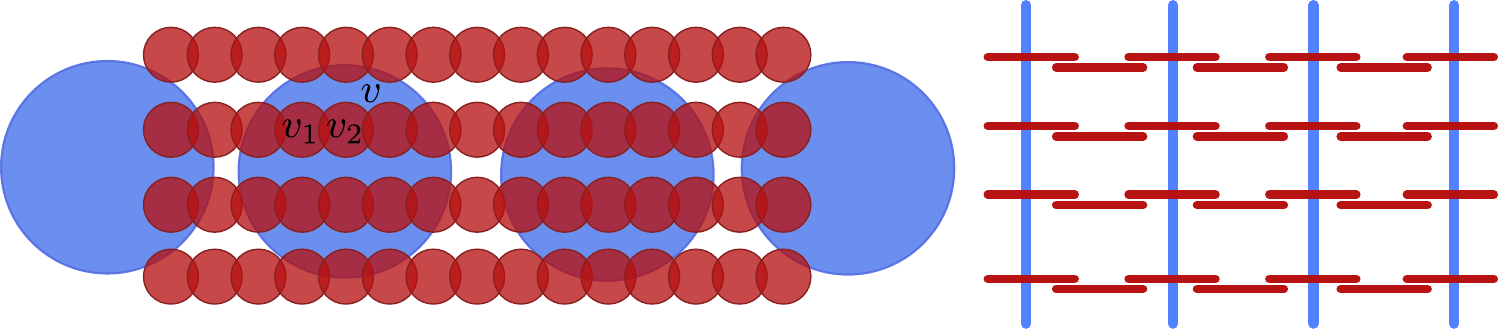}
    \caption{Example of a $K_{r,r}$ contained as a minor for $r=4$ in a disk graph (left) and a $2$-DIR graph (right). In the case of disk graph, $v$ has a matching of size $r-2$ in its neighborhood, forming a triangle bundle, which can be exploited to branch. The set $M$ are depicted in blue. For the $2$-DIR graph, the vertices of the long paths are represented by segments with small variation in their height and not intersecting for better clarity, but are in fact on the same level and intersecting. }
    \label{fig:introKrr}
\end{figure}

A point that seems crucial to us is the following.
In intersection graphs of ``fat objects'' (like disks, squares, or pseudo-disks more generally), 
the ``locally non planar structure'' when an object (vertex $v$ in \autoref{fig:introKrr}) is ``traversed'' (by $v_1$, $v_2$ in \autoref{fig:introKrr})
comes to the price of an edge ($\{v_1,v_2\}$) in the neighborhood of $v$.
Thus, the presence of a large $K_{r,r}$ as a minor implies that a large matching $E_v$ (of size $\Omega(r)$) will appear in the neighborhood of a vertex $v$.
However, as $G[\{v\} \cup E_v]$ (called a triangle bundle in \cite{Lokshtanov23Approx}) contains $r$ triangles pairwise intersecting on exactly one vertex $v$, the set $\{v\} \cup E_v$ is a good structure to perform a subexponential branching for \FVS. Indeed, \cite{lokSODA22} proposed a ``virtual branching'' to handle this structure by either taking $v$ in the solution, or absorbing $E_v$ in $M$, implying then that the parameter virtually decreases by $|E_v|$ as a solution which does not contain $v$ has to hit all these edges, even if we cannot branch to determine which are exactly the vertices in the solution.

Once no more virtual branching is possible on large triangle bundles, they obtain by some additional specialized techniques that any vertex in $M$ is such that $N_{V(G-M)}(v)$ is an independent set. Then, it is proved (~\cite{lokSODA22}, Corollary 1.1) that in a disk graph where for any $v \in M$, $N_{V(G-M)}(v)$ is an independent set, and where there does not exist a vertex in $V(G-M)$ whose neighborhood is
contained in $M$, then $\tw(G)=\O\left(\sqrt{|M|}\omega(G)^{2.5}\right)$, where $\omega(G)$ denotes the maximum size of a clique in $G$. This no longer holds for $2$-DIR graphs: the family of pairs $(G,M)$ depicted on the right of \autoref{fig:introKrr} is indeed a counter example as they respect the conditions, have $\omega(G)=2$, but $\tw(G)=\Omega(|M|)$.

More generally, the role of the size of a matching in the neighborhood was studied in~\cite{berthe24ASQGM} which shows how subexponential parameterized algorithms can be obtained for graph classes having the ``almost square grid minor property'' (ASQGM), corresponding more or less\footnote{In the correct definition $\mu_N(G)$ is replaced by a slightly more technical parameter.} to $\tw(G) = \O(\omega(G)^{\O(1)}\mu_N(G)\grid(G))$ where $\mu_N(G)$ is the maximum size of a matching in a neighborhood of a vertex, and $\grid(G)$ is the largest size of a grid contained as a minor in $G$.
The previous counterexample shows that $2$-DIR does not have the ASQGM property, implying that we need another approach to handle them.

\subsection{A simpler case study: when trees are only paths}\label{sec:intro:paths}
To simplify the arguments, but still understand why properties in \autoref{def:prop} of a nice graph class are needed, let us assume that the forest $G-M$ only contains paths $(P_i)_i$, and consider the example of segment graphs.
This case remains challenging as a large $K_{r,r}$ can still be hidden as a minor (as in \autoref{fig:introKrr}), and we need to destroy it in order to reduce the treewidth of the graph. 
To keep notations simple, we use the notation $\poly(.)$ to denote a polynomial dependency on the parameter, and thus we do not try to compute tight formulas depending on the polynomial $f_1,f_2,f,d$ given in the definition of a nice class.
We assume that we performed folklore branching and that we are left with a $K_{r,r}$-free graph $G$ for $r=k^{\eps}$.
It is known for string graphs (and thus segment graphs) that in this case $\tw(G)=  \poly(k^{\eps})n^{1/2}$ (corresponding to \autoref{def:prop:tw} of the definition of nice). Thus, our goal is to reduce $|V(G)|$ to $\O(k^{2-\eps'})$ for some $\eps'$.
A first obvious rule is to iteratively contract edges of the $P_i$'s whose endpoints have no neighbors in $M$ (corresponding hereafter to \ruleref{rl:deg2}). This explains the property of \autoref{def:prop:deg2} of the definition of nice. Note that after if this rule does not apply, at least half of the vertices of the $P_i$'s have neighbors in $M$.

Let us now explain how 
this rule and another rule, lead to the following ``degree-related size property'': for any subpath $P$ of a $P_i$, $|P| = \O(d_M(P)^2)$. Consider the subset of vertices of $P$ having neighbors in $M$, and note that it has size at least roughly $|P|/2$.
This implies that there are at least $|P|/2$ edges between $P$ and $N_M(P)$.
We now need a new kernelization rule (corresponding to \ruleref{rl:neighborhood}): If there exists a vertex $u\in M$  adjacent to $x$ vertices in $P$, for $x\approx 2 d_M(P)$, then it is optimal to force $u$ in the solution. After applying this rule, for any subpath $P$ of a $P_i$, every vertex of $N_M(P)$ is adjacent to at most $2 d_M(P)$ vertices in $P$, hence there are at most $2 d_M(P)^2$ edges between $N_M(P)$ and $P$. Thus $|P|/2 \le 2 d_M(P)^2$.

Before the next step we need to apply the following ``large degree rule'' (corresponding to \ruleref{rl:big}): if there is a vertex $v$ in a $P_i$
such that $d_M(v) > t$, then add $v$ to $M$.
One can prove that by taking $t=2\d(r)$, with $\d(r)=\poly(r)$ the constant defined in \autoref{def:prop:sparse} of the definition of a nice class, $M$ does not grow too much after applying this rule exhaustively: by denoting $A\subseteq M$ the set of vertices already in $M$ before applying this rule, we always have $|A|=\poly(r)|M\setminus A|$. This claim will be discussed in \autoref{ssec:path2tree}.

Observe that at this stage we may still have a large $K_{r,r}$ as a minor, with for example graphs as in \autoref{fig:introKrr} where no rule applies.
It remains to define a crucial rule to destroy these $K_{r,r}$-minors.
Let us now present an algorithm that partitions the $P_i's$, whose analogue in the general case is called \algopart in \autoref{sssec:partition-algo}.
For any connected component $P_i$ in $G-M$, we start (see \autoref{fig:introPart}) from an endpoint of $P_i$ and collect greedily vertices until we find a subpath $P_i^1$ such that $d_M(P_i^1) \ge t$, or that there is no more vertices in $P_i$. If $d_M(P_i^1) \ge t$, then restart a new path starting from the next vertex to create $P_i^2$, and so on. Call $x(P_i)$ the number of subpaths defined this way from $P_i$.
This defines a partition $P_i=\bigcup_{1\leq \ell \leq x(P_i)}(P_i^\ell)$, where $d_M(P_i^\ell) \ge t$ for any $\ell \in [1,x(P_i)-1]$ and no lower bound for $d_M\left(P_i^{x(P_i)}\right)$.
As we applied the large degree rule, we also know that $d_M(P_i^\ell) \le 2t$ for any $\ell \in [1,x(P_i)]$, because collecting at each step a new vertex in the path $P_i^\ell$ can increase $d_M(P_i^\ell)$  by at most $t$.
This implies, using the degree-related size property introduced above, that $|P_i^\ell| = \O(t^2) \le \poly(r)$ for any $\ell \in [1,x(P_i)]$.
Let us denote $\mT^+$ the set of those $P_i^\ell$'s such that $d_M(P_i^\ell) \ge t$, i.e., the ``large-degree subpaths''. Observe that the last considered subpath $P_i^{x(P_i)}$ of each connected component $P_i$ (on the right of each path of $G-M$ in \autoref{fig:introPart}) may have $d_M\left(P_i^{x(P_i)}\right)<t$ as there was no more vertices to complete it, hence it is not contained in $\mT^+$. We denote  $\mT^-$ those remaining ``small-degree subpaths''.

\begin{figure}[!ht]
    \centering
    \includegraphics[scale=1.2]{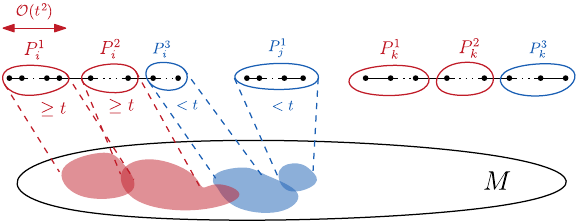}
    \caption{Example of partition where $P_i$ is partitioned into $x(P_i)=3$ subpaths, with $P_i^1$ and $P_i^2$ in $\mT^+$ and $P_i^3 \in \mT^-$.}
    \label{fig:introPart}
\end{figure}

\begin{figure}[!ht]
    \centering
    \includegraphics[scale=1.2]{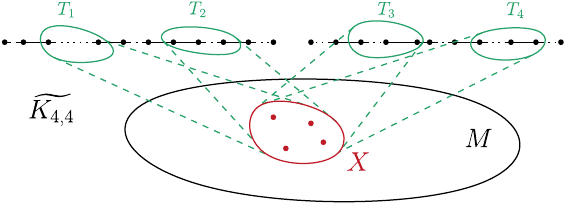}
    \caption{Example of $\KttTilde$ for $t=4$. Here we have $X\subseteq N_M(T_i)$ for $1\leq i \leq 4$.}
    \label{fig:introKtttilde}
\end{figure}

Let us now explain how the bounded tree neighborhood complexity property (\autoref{def:NCtrees:item2}) allows to obtain 
the following ``small number of large-degree subpaths'' property. By removing half of the $\mT^+$'s, we can get a set $\mT^{+'}$ of non-adjacent trees (meaning with no edge between the $T_i's$) such that $|\mT^{+'}|\geq \frac 12|\mT^+|$. We can then apply the bounded tree neighborhood complexity property with $A=M$, $\mT=\mT^{+'}$ and $p=m=\poly(r)$ to obtain
$|\{N_M(T),~T \in \mT^{+'} \}| \leq \poly(r)|M|$.
Thus, if $|\mT^{+'}| \ge x \cdot \poly(r)|M|$, we found $x$ large-degree subpaths (denoted $T'_i$) having the same neighborhood $X'$ in $M$ with $|X'|\geq t$.  By choosing $x=t$ and considering $X\subseteq X'$ with $|X|=t$, we found a structure that we call a $\KttTilde$ (see \autoref{fig:introKtttilde}), formally defined as a pair $\left(X, \{T_i\}_{1\leq i \leq t}\right)$ where
\begin{itemize}
    \item $X\subseteq M$ has size $t$, 
    \item  $\{T_i\}_i$ a family of $t$ vertex-disjoint non-adjacent subtrees (paths here) of $G - M$ such that for all $1\leq i \leq t$, $X\subseteq N_M(T_i)$, and
    \item for any $T_i$, $|T_i| \le \poly(r)$.
\end{itemize}
Now, inspired by the aforementioned ``virtual branching rule'' of \cite{Lokshtanov23Approx} for triangles bundles, we introduce a branching rule (corresponding to \ruleref{rl:KrrTilde} hereafter) that either deletes almost all vertices in $X$, or is adding to $M$ a subset of $t-1$ of the $T_i's$. The complexity behind this rule is fine, as in the second branch, the parameter virtually decreases by $t-1$, and $M$ grows by $(t-1)\max_i\{|T_i|\} = \poly(r)$.
This explains how we deal with the $K_{t,t}$-minors.

Finally, if this $\KttTilde$ rule cannot be applied, it remains to bound $|V(G-M)|$. 
Recall that any $P \in \mT^+ \cup \mT^-$ is such that $|P| \le \poly(r)$, and thus we only need to bound $|\mT^+ \cup \mT^-|$.
As we cannot apply the previous rule, we know that the number of big paths is small: $|\mT^+| \le \poly(r)|M|$. Now, to bound 
$|\mT^-|$, observe that we can partition $\mT^-=\mT^-_1 \cup \mT^-_2$, where
\begin{itemize}
    \item $\mT^-_1$ is the set of small-degree paths $P_i^\ell$ for some $\ell > 1$ (belonging to the same path $P_i$ than a large-degree path $P_i^{\ell-1}$).
    \item $\mT^-_2$ is the set of small-degree paths which are entire connected components of $G-M$.
\end{itemize}
As $|\mT^-_1| \le |\mT^+|$, it only remains to bound $|\mT^-_2|$.

Let us now explain how to obtain  $|\mT^-_2| \le \poly(r)|M|$.
According to \autoref{def:NCtrees:item2} of the bounded tree neighborhood complexity property with $\mT = \mT^-_2$, $A=M$, $p=t$, and $m=\poly(r)$, we get that  
$|\{N_M(T), T \in \mT^-_2 \}| \le f_2(r,p,m)|M|$. 
Thus, by pigeonhole principle, if 
$|\mT^-_2| \ge (t+2)f_2(r,p,m)|M|$, then there exists $t+2$ paths of $\mT^-_2$ having the same neighborhood $X$ in $M$.
This case corresponds to a
$\widetilde{K_{t_1,t_2}}$ where $t_1 = |X|$ and $t_2 = t+2$. However, as $t_1$ may be arbitrarily small, we cannot branch as we did on $\KttTilde$ 
as the branches where we delete almost all vertices of $X$ do not decrease the parameter $k$ by a large amount. 
However, in this case, paths of $\mT^-_2$
 are just connected components of $G-M$, and this helps to obtain a last rule (\ruleref{rl:same_neighborhood_prime}) that identifies paths that can be safely removed.
 This rule, restated in our context, says that if $|X|+2$ paths in $G \setminus M$ have the same neighborhood $X$ in $M$, then one of them can be removed. As in our case we found $t+2$ paths having the same neighborhood $X$, and $t > |X|$, we can indeed apply the rule.
Thus, once \ruleref{rl:same_neighborhood_prime} cannot be applied, we get 
$|\mT^-_2| < (t+2)f_2(r,p,m)|M| \le \poly(r)|M|$.

 This concludes the sketch of proof for this restricted setting where the connected components of $G-M$ are paths, as we obtain by taking $\eps$ small enough
 $|V(G-M)| \le \poly(r)|M| \le \poly(k^{\eps})k= \O(k^{2-\eps'})$ as required.
 
 \subsection{Challenges to lift the result from paths to trees}\label{ssec:path2tree}
 We now consider the general setting where given an instance $(G,k)$, with $G$ being $K_{r,r}$-free for $r=k^\eps$, and given $M$ a feedback vertex set of size at most $2k$, we want to reduce the graph to obtain $|V(G-M)|= \O(k^{2-\eps'})$. The approach still consists in partitioning $G-M$ in an appropriate way (called a $t$-uniform partition).

 A first problem when trying to adapt the approach described in \autoref{sec:intro:paths} above is the degree-related size property. Indeed, after the first two sections \autoref{sec:propertiesKernelized} and \autoref{sec:bigT}, we are now only able to obtain that for any subtree $T$ of $G-M$,  $|T| \le \poly(r)\db(T)^{\O(\alpha)}$ where $\db(T)=\max(d_M(T),\bom(T))$ and $\bom(T)=|\{v\in T,~N(v) \not\subseteq M\cup T \}|$ is the size of the ``border of $T$''.  Observe that $\bom(P)$ is at most $2$ for any subpath $P$ of path $P_i$, whereas $\bom(T)$ can only be bounded by $|T|$ for a subtree $T$.
 Informally, in the path case $|P|$ could be bounded by a polynomial function of $d_M(P)$ only, while now
 $|T|$ also depends (polynomially) on $\bom(T)$.
 
 A second problem is the large degree rule. 
 Suppose that this rule no longer applies (meaning that for every $u\in V(G-M)$, we have $d_M(u)\leq t$), and suppose now that because of another rule a vertex $v \in V(G-M)$ is added to $M$, denoting $M'=M \cup \{v\}$. Then this can create a new large degree vertex 
 $v'$ with $d_{M'}(v')>t$ (and so $d_{M}(v') = t$). Then $v'$ would need to be added and the problem may arise again for another vertex $v''$. This ``cascading'' could easily be prevented if $G-M$ is a union of paths: it suffices to apply the rule a first time at the start of the algorithm, but with $t'=t-2$. We then have for each $v\in V(G-M)$ the bound $d_M(v)\leq t-2$, and we do not need to apply the rule again after as adding vertices to $M$ may increase $d_M(v)$ by at most $2$ ensuring the wanted bound $d_M(v)\leq t$ for $v\in V(G-M)$.
 However, in the case when $G-M$ is a forest, it may contain a vertex of arbitrarily large degree, so we cannot apply the same solution.
 The problem is treated with the help of a technical lemma (that we prove at the end of the proof, see \autoref{lm:sizeM}) which ensures that throughout the execution of the algorithm we keep $|M|=\poly(r)k$.
 
 Finally, a third problem is the definition of the partition. As in the case of paths we want to partition $G-M$ into a ``$t$-uniform partition'' $\mT$, where in particular we have $\mT=\mT^+ \cup \mT^-$, and for any $T \in \mT$, $d_M(T)\le 2t$ and $|T|\leq \poly(r)\db(T)^{\O(\alpha)}$ (see \autoref{def:unif} for the complete definition).
 The greedy approach presented for the case of paths is now more involved, as we have to cut each tree of $G-M$ into subtrees that have small border $\bom(T)$, as otherwise the previous bound $|T| \le \poly(r)\db(T)^{\O(\alpha)}$ becomes useless when $\db(T)$ is too large.
 
 This partitioning procedure is defined in \autoref{sssec:partition-algo}. 
 It can either: 
 \begin{itemize}
     \item Fail and find a subtree $T$ with $|T|>\poly(r)\db(T)^{C\alpha}$ for some constant $C$, implying that our degree-related size rule can be applied.
     \item Fail and find too many subtrees $T_i \in \mT^+$ with large degree, implying that we found a $\KttTilde$, and that our $\KttTilde$ rule can be applied.
     \item Produce a $t$-uniform partition with $|\mT^+| \le \poly(r)|M|$.
 \end{itemize}
 The third case is treated in \autoref{ssec:step4}, where we either find another way to apply one more time a reduction rule, or prove that $|V(G-M)|= \O(k^{2-\eps'})$.

\section{FVS in subexponential FPT time in nice graph classes}\label{sec:algo}

\subsection{Preliminary branching to remove \texorpdfstring{$K_{r,r}$}{Krr}}\label{ssec:Krr_branch}
To avoid confusion, we refer to the initial instance with $(G_0,k_0)$. In this section we use a folklore branching for \FVS to remove the large bicliques $K_{r,r}$, where $r=k_0^{\eps}$ with $\eps$ to be set later depending on the considered graph class $\G$.
Before performing any branching, we compute a 2-approximation of a minimum feedback vertex set of $G_0$ using the following result, and denote it by~$M_0$.

\begin{theorem}[\cite{2approx-fvs-1,2approx-fvs-2}]\label{th:fvsapprox}
    A $2$-approximation of a minimum feedback vertex set can be constructed in polynomial time.
\end{theorem}

If $|M_0|\le k_0$ or if $|M_0|> 2k_0$ we can immediately conclude that the instance is positive or negative. We thus are left with the case where $k_0 < |M_0| \le 2k_0$.

Let us now describe a branching routine
leading to a new set of instances $\I$, whose properties are discussed below.
The routine initializes $\I$ to $\{(G_0, k_0, M_0)\}$ and applies the following branching rule to the elements of $\I$ as long as possible.

\begin{enumerate}[label=(\BR{{\arabic*}})]
\setcounter{enumi}{0}
\item \label{rl:Krr} Given an instance $(G, k, M)\in \I$, if $G$ contains a $K_{r,r}$-subgraph with parts $A$ and $B$,
either $k < r - 1$ and the instance is negative so we remove it from $\I$, or $k\geq r-1$ and we replace the considered instance by the $2r$ instances of the form $(G -  X, k-(r-1), M\setminus X)$, for every set $X$ of size $r-1$ that is either subset of $A$, or subset of $B$.
\end{enumerate}

\begin{lemma}\label{lem:findkrr}
   There is a  $2^{\O\left (r\log |M| \right )}|V(G)|^{\O(1)}$-time algorithm that, given an instance $(G, k, M)\in \I$, applies rule~\ruleref{rl:Krr} or correctly concludes that $G$ is $K_{r,r}$-free.
\end{lemma}
\begin{proof}
    As $G-M$ is acyclic, for any copy of $K_{r,r}$ with parts $A$ and $B$, one part of this copy is almost fully contained in $M$. More formally and without loss of generality we can assume that $|A\cap M|\ge r-1$.
    Therefore, to iterate over all possible choices of $A$ we can consider all the subsets of $V(G)$ consisting of $r-1$ vertices of $M$ and adding to this a vertex of $V(G)$. For each such choice we can consider each vertex $v$ of $V(G)\setminus A$ and check if $A\subseteq N(v)$. If there are $r$ such vertices, then together with $A$ they for a copy of $K_{r,r}$. It is then trivial to produce the $2r$ instances described in \ruleref{rl:Krr}. All this takes $\O(|M|^r) \cdot |V(G)|^{\O(1)}=2^{\O(r\log |M|)}|V(G)|^{\O(1)}$ steps.
\end{proof}

We summarize the properties obtained after this series of branchings with the following lemma:
\begin{lemma}\label{lem:step2}
    At the end of the preliminary branching, the set $\I$ satisfies the properties:
    \begin{enumerate}
        \item The instance $(G_0, k_0)$ is a \textsc{Yes}-instance if and only if $\I$ contains a $YES$-instance.
        \item For any $(G,k,M)\in \I$ the graph $G$ is a $K_{r,r}$-free induced subgraph of $G_0$.
        \item For any $(G,k,M)\in \I$, $M$ is a feedback vertex set of $G$ with $|M|\leq 2k_0$.
        \item The total time to generate $\I$ is in $2^{\O\left (r\log k_0 \right )}|V(G_0)|^{\O(1)}(2r)^{\frac{k_0}{r-1}}$
        and $|\I|=\O\left((2r)^{\frac{k_0}{r-1}}\right)$.
    \end{enumerate}
\end{lemma}
\begin{proof}
The first item follows from the fact that if rule~\ruleref{rl:Krr} applies to a triplet $(G,k,M)$, and replaces it with a set of $2r$ instances, we have the property that $(G,k,M)$ is a \textsc{Yes}-instance if and only if (at least) one of the $2r$ instances is. In order to show this let us call $A$ and $B$ the vertex sets of the parts of the considered $K_{r,r}$-subgraph.

If $(G,k,M)$ is a \textsc{Yes}-instance, let $S$ be a feedback vertex set of $G$ with $|S|\le k$. Any 2 vertices of $A$ induce a 4-cycle with any 2 vertices of $B$, so in order to intersect all cycles $S$ intersects one of $A$ and $B$ on at least $r-1$ vertices. Suppose without loss of generality that $S$ intersects $A$ on (at least) $r-1$ vertices, let $a\in A$ be the $r^{\text{th}}$ vertex of $A$ (i.e. $a$ is a vertex such that $A\setminus\{a\} \subseteq S$, but possibly $a\in S$), and let $X=A\setminus \{a\}$. Let $(G',k',M')$ be the instance, among the $2r$ ones generated by rule~\ruleref{rl:Krr}, such that $G'=G -  X$, and $M'=M\setminus X$. Note that taking $S'=S\setminus X$, we have that $G -  S=G' -  S'$ and $|S'|=|S|-(r-1)\le k-(r-1)=k'$. Hence, $(G',k',M')$ is a \textsc{Yes}-instance.

Conversely, if an instance $(G',k',M')$ generated by \ruleref{rl:Krr} is a \textsc{Yes}-instance, with $G'=G -  X$, and $M'=M\setminus X$, let $S'$ be a feedback vertex set of $G'$ with $|S'|\le k'=k-(r-1)$.
Note that taking $S=S' \cup X$, we have that $G' -  S' = G -  S$ and $|S|=|S'|+(r-1)\le k'+(r-1)=k$. Hence, $(G,k,M)$ is a \textsc{Yes}-instance.

For the second item, if there was a copy of $K_{r,r}$ then the branching process would not be finished. Furthermore, each time rule~\ruleref{rl:Krr} is applied, the new instances are obtained by deleting vertices from the instance formerly in $\I$. Hence, all the generated instances are induced subgraphs of the first one, $G_0$.

For the third item, this follows from the fact that each time rule~\ruleref{rl:Krr} is applied, the new sets $M'$ are subsets of $M$, and thus subsets of $M_0$, which is of size at most $2k_0$.

For the last item, let us consider a recursive algorithm $\tilde{A}$ that given $(G,k,M)$ as input, and using \ruleref{rl:Krr}, output all instances of $\I$ generated from $(G,k,M)$ (implying that $\tilde{A}(G_0,k_0,M_0)=\I$).
By \autoref{lem:findkrr}, and by observing that for each call of the algorithm we have $|M|\leq 2k_0$, the worst case running time $f(n,k)$ of $\tilde{A}$ on a $n$-vertex graph and parameter $k$ is such that $f(n,k) \le 2^{\O\left (r\log k_0 \right )}n^{\O(1)}+(2r) f(n,k-(r-1))$.
This implies $f(n,k) \le 2^{\O\left (r\log k \right )}n^{\O(1)}(2r)^{\frac{k}{r-1}}$.
Moreover, as the number $g(k)$ of instances
generated by a call to $\tilde{A}(G,k,M)$ is such that $g(k) \le (2r) g(k-(r-1))$, we obtain the claimed bound on $|\I|$.
\end{proof}

\subsection{The main recursive algorithm}
\label{ssec:step3}

We now consider each element $(G_i,k_i,M_i)$ of  $\I$ as an instance $(G_i,k_i,M_i,\emptyset)$ of the following auxiliary problem \AFVS, and our goal now is to solve these instances of \AFVS using the recursive \autoref{algo:A} that we described hereafter.

\begin{definition}\label{def:afvs}
    Given a nice graph class $\G$ and an integer $r>3$, the \afvs problem (\AFVS{} for short) is the decision problem that takes as input a 4-tuple $(G, k, M, \mH)$ where:
    \begin{itemize}
        \item $G$ is a $K_{r,r}$-free graph of $\G$,
        \item $k$ an integer,
        \item $M\subseteq V(G)$ a feedback vertex set of $G$, and
        \item $\mH$ a family of at most $k$ connected\footnote{That is, for any $H \in \mH$, $G[H]$ is connected.} disjoint subsets of $M$,
    \end{itemize}
    and where the question is to decide whether there is a feedback vertex set $S$ of $G$ of size at most $k$ that additionally intersects every set of $\mH$.
\end{definition}

For the sake of completeness we provide here the complete pseudo-code of \autoref{algo:A}, even if it uses rules and routines which will be defined later. At this stage, we recommend the reader to only read the following sketch, as the following sections will cover in detail the properties we obtain after each step.
The sketch of \autoref{algo:A} is as follows. We first try to apply (line~\ref{A:firstIf}) rules~\ruleref{rl:deg1}, \ruleref{rl:deg2}, \ruleref{rl:big}, and \ruleref{rl:same_neighborhood_prime},  which are like kernelization rules: given the instance $(G,k,M,\mH)$ we perform a single recursive call on a slightly simpler instance $(G',k',M',\mH')$. 
 If none of these first rules apply, the algorithm tries to build a special partition of $G-M$ using the routine \algopart{} (defined in \autoref{lem:construct-red-trees}).
 If \algopart fails (line~\ref{A:case1} or \ref{A:case2}) and fall into what we call Case 1 or Case 2, then we apply a kernelization or branching rule.
 Otherwise we either apply \ruleref{rl:same_neighborhood_prime} or \ruleref{rl:neighborhood} (line~\ref{A:KR4} or line~\ref{A:KR5}), or reach our final point (line~\ref{A:DP}) where we can prove that $|V(G)|$ is small, implying that $\tw(G)=\O(k^{1-\eps'})$, and solve the instance using a classical dynamic programming algorithm.

\begin{algorithm}
\caption{$A(G,k,M,\mH)$}
\begin{algorithmic}[1]
    \Require $(G,k,M,\mH)$ an instance of \AFVS.
    \If{(one of Rule~\ruleref{rl:deg1}, Rule~\ruleref{rl:deg2}, Rule~\ruleref{rl:big}, or Rule~\ruleref{rl:same_neighborhood_prime} applies on $(G,k,M,\mH)$)} \label{A:firstIf}
         \State Apply the first possible Rule to obtain $(G', k', M',\mH')$ 
        and \Return $A(G',k',M',\mH')$
    \EndIf
    \State Apply \algopart (with $t=2\d(r)$) of \autoref{lem:construct-red-trees} that tries to build $\mT$: a $t$-uniform partition of $G-M$ with $|\mT^+| \le p_3(r,t)|M|$ (where $p_3$ is defined in \autoref{lm:bigreduce}) \label{A:partition}
    \If{(procedure fails and falls into Case 1 (output a large subtree $T$))} \label{A:case1}
        \State Apply Rule~\ruleref{rl:neighborhood} on $T$ to obtain $(G', k', M',\mH')$
        and \Return $A(G',k',M',\mH')$
    \EndIf
    
    \If{(procedure fails and falls into Case 2 (output a $\KttTilde$))}
     \label{A:case2}
        \State Apply the branching Rule~\ruleref{rl:KrrTilde} on this $\KttTilde$, generating a set $\mathcal C$ of instances
        \State \Return $\bigvee_{(G',k',M',\mH') \in \mathcal C}A(G',k',M',\mH')$
    \EndIf
    \State // $\mT$ is as required  \label{A:partOK}
    \State Let $Z_1(\mT)$ and $Z_2(\mT)$ as defined in \autoref{def:Z}, and $\tilde{M}=M \cup Z_1(\mT) \cup Z_2(\mT)$
    \State // By \autoref{lm:special}, Rule~\ruleref{rl:deg1} and Rule~\ruleref{rl:deg2} do not apply on $(G,k,\tilde{M},\mH)$
    \If{(Rule~\ruleref{rl:same_neighborhood_prime} applies on $(G,k,\tilde{M},\mH)$, and finds a subtree $T$ that can be removed)}\label{A:KR4}
    \State\Return $A(G-T,k,M,\mH)$
    \EndIf
    \If{(Rule~\ruleref{rl:neighborhood} applies on $(G,k,\tilde{M},\mH)$ and a connected component $T$ of $G-\tilde{M}$, and finds a vertex $u \in \tilde{M}$ that can be taken)}\label{A:KR5}
    \State\Return $A(G-T,k-1,M\setminus \{u\},\mH-\{u\})$ 
    \EndIf

    \State // $|V(G)|=p_4(r,t)|M|$ by \autoref{lm:finalR4R5} implying $\tw(G)=\O(k^{1-\eps'})$ by \autoref{th:tw}
    \State \Return $DP(G,k,M,\mH)$ \label{A:DP} // Solves the instance using \autoref{th:solveAFVS}
\end{algorithmic}
\label{algo:A}
\end{algorithm}

\subsection{Kernelization rules}
\label{sssec:kernelization}

Here we provide the four kernelization rules  \ruleref{rl:deg1}, \ruleref{rl:deg2}, \ruleref{rl:big}, and \ruleref{rl:same_neighborhood_prime} that \autoref{algo:A} tries to apply at line~\ref{A:firstIf}.
Each of these rules takes as input an instance $(G,k,M,\mH)$ of \AFVS and outputs a single instance $(G',k',M',\mH')$. Such a rule is said to be {\em safe} if:
\begin{itemize}
    \item For an instance of \AFVS the rule returns an instance of \AFVS (in particular the graph of the output instance is a $K_{r,r}$-free graph in $\G$), and
    \item the input instance is a \textsc{Yes}-instance if and only if the output instance is a \textsc{Yes}-instance.
\end{itemize}

\paragraph{Notation}
Given any instance $(G, k, M, \mH)$, recall that every connected component of $G-M$ is a tree. We root each of them at an arbitrary vertex. We define a {\em subforest} of $G-M$ as a subset $T\subseteq V(G-M)$ and say the set is a {\em subtree} of $G-M$ if $G[T]$ is a tree. Given a vertex $v$ of a connected component $T$ in $G-M$, we define the subtree $T_v$ of $G-M$ as the connected component of $v$ in $G-M-u$, where $u$ is the parent of $v$, if any. If $v$ is the root of $T$, then $T_v=T$. In any case $T_v$ is rooted at~$v$.
Given $X \subseteq V(G)$, we denote 
$\mH-X = \{H \in \mH \mid H \cap X = \emptyset\}$. 
 Given a subtree $T$ of $G-M$, $\pom (T)$ denotes the set $\{v\in T,~N(v) \not\subseteq M\cup T \}$, and $\bom (T)$ denotes the size of this set.
 We also denote $\db(T)=\max(d_M(T), \bom(T))$.

We start with two basic reduction rules often used to deal with \FVS and that allow to get rid of vertices of degree 1 and arbitrarily long paths of vertices of degree 2. Notice that we only apply here the reduction to vertices in $G-M$.
\begin{enumerate}[label=(\KR{{\arabic*}})]
\item \label{rl:deg1} Given an instance $(G, k, M, \mH)$, if there exists a vertex $v\in V(G-M)$ of degree $d(v)\leq 1$, output $(G-\{v\}, k, M, \mH)$.

\item \label{rl:deg2} Given an instance $(G, k, M, \mH)$, if there exists a path $quvw$ in $V(G-M)$ such that the four vertices have degree $2$ in $G$ (hence $\dM(u)=\dM(v)=0$), output $(G', k, M, \mH)$, where $G'$ is the graph obtained from $G$ by contracting the edge $uv$.
\end{enumerate}
\begin{lemma}
    The rules \ruleref{rl:deg1} and \ruleref{rl:deg2} are safe and can be applied in polynomial time.
\end{lemma}
\begin{proof}
Notice that the vertices considered in the two rules belong to $G-M$, in particular they do not belong to any member of $\mH$ (which are subsets of $M$). Contracting an edge with endpoints of degree two preserves being $K_{r,r}$-free and does not modify the size of the minimum feedback vertex sets. Moreover the obtained graph is still in $\G$ by definition of a nice class.
The output is then an instance of \AFVS equivalent to the input. The running time claim is immediate.

\end{proof}
\begin{remark}
Observe that we did not use the condition in \ruleref{rl:deg2} that the endpoints of the considered path have degree at most $2$ in $G$. However it will later be used in \autoref{lm:sizeM} to bound the size of $M$ during the execution of the algorithm.
\end{remark}

The next rule ensures that the vertices outside $M$ have a small neighborhood in $M$. It increases the size of $M$, but in a controlled manner as we will see later in \autoref{lm:sizeM}. 

\begin{enumerate}[label=(\KR{{\arabic*}})]
\setcounter{enumi}{2}
\item \label{rl:big} Given an instance $(G, k, M, \mH)$, if there is a vertex $v\in V(G-M)$ such that $d_M(v)\ge t=2\d(r)$, with $\d(r)$ the value defined in \autoref{def:prop}, output $(G, k, M\cup\{v\}, \mH)$.
\end{enumerate}
It is immediate that rule~\ruleref{rl:big} is safe and can be applied in polynomial time.

The fourth kernelization rule deletes unnecessary trees of $G-M$.
Consider an instance $(G,k,M,\mH)$ where none of the previous rules applies. 
Given a family $\T$ of disjoint subtrees of $G-M$, a tree $T\in \T$ is {\em redundant} (for $\T$) if for all $v\in M$ such that $d_T(v)\geq 2$, there exists $T'\in \T$ with $T'\neq T$ such that $d_{T'}(v)\geq 2$.

\begin{lemma}\label{lm:redundant}
    Consider a set $X\subseteq M$ and a set $\T$ of more than $|X|$ subtrees of $G-M$ such that for every $T\in \T$,  $N_M(T)=X$. Then, there exists a redundant tree for $\T$, and it can be found in polynomial time.
\end{lemma}
\begin{proof}
    For any $x\in X$, if there is at least one subtree $T\in \T$ such that $d_T(x)\geq 2$, we arbitrarily pick one of them and call it $T^{(x)}$. Now, let $T'$ be one of the remaining trees in $\T$. Such a tree exists as $|\T|\geq |X|+1$, and there are at most $|X|$ trees of the form $T^{(x)}$, for some $x\in X$. It is immediate from the definition that $T'$ is redundant.
\end{proof}

Recall that we consider a fixed child-parent orientation in the forest $G-M$. In what follows, we say that a subset $F \subseteq V(G-M)$ is a \emph{downward-closed subtree} (or \emph{subforest} when the set is not necessarily connected) of $G-M$ when for any $v \in F$ and any children $u$ of $v$, $u \in F$.
The fourth kernelization rule is as follows:

\begin{enumerate}[label=(\KR{{\arabic*}})]
\setcounter{enumi}{3}
\item \label{rl:same_neighborhood_prime} Given an instance $(G, k, M, \mH)$, a set $X\subseteq M$ with $|X|\geq 1$ and a set $\T$ of at least $|X|+2$ disjoint downward-closed subtrees of $G-M$ such that:
\begin{itemize}
    \item for all $T\in \T$, we have $N_M(T)=X$, and
    \item either all the roots of the trees in $\T$ have a common parent $r$, or $\T$ consists only in connected components of $G-M$,
\end{itemize}
arbitrarily pick one 
redundant $T\in \T$ (which exists as shown in \autoref{lm:redundant}) and output $(G - V(T), k, M, \mH)$.
\end{enumerate}

\begin{lemma}\label{lm:applyKR4}
    Rule~\ruleref{rl:same_neighborhood_prime} is safe and can be applied in polynomial time.
\end{lemma}

\begin{proof}
The fact that the rule can be applied in polynomial time directly follows from \autoref{lm:redundant}.
Also, as the graph $G'$ obtained after applying the rule is an induced subgraph of $G\in \G$, we have $G'\in \G$.

Let us now show that the input and output instances are equivalent.
As noted above $G'$ is an induced subgraph of $G$ so any feedback vertex set of $G$ is also one for $G'$.
Hence we only have to show that if $G'$ has a feedback vertex set $S'$, then $G$ has a feedback vertex set of size $|S'|$ too. Consider such a set $S'$ and let us transform it (if needed) into a feedback vertex set $S$ of $G$ according to the following cases.

\begin{description}
    \item[First case] $|X\setminus S'|\geq 2$. In this case we use the following claim.
\begin{claim}
    If $|X\setminus S'|\geq 2$, then $S'$ intersects all the trees of $\T$ except $T$ and at most one other.
\end{claim}
\begin{proof}
Let $x_1,x_2\in X\setminus S'$ be distinct vertices, and suppose by contradiction that there are distinct $T_1,T_2\in \T\setminus \{T\}$ such that $S$ does not intersect $T_1$ nor $T_2$. There exist $a_1\in N_{T_1}(x_1)$ and $b_1\in N_{T_1}(x_2)$ as $N_M(T_1)=X$, and because $T_1$ is connected there exists a path from $a_1$ to $b_1$ in $T_1$ (observe that we may have $a_1=b_1$). Similarly, we define $a_2, b_2$ and a path joining them in $T_2$. Then we have a cycle in $G'$ that is not hit by $S'$, a contradiction.\cqed
\end{proof}

So if $|X\setminus S'|\geq 2$, there are at least $|X|$ vertices of $S'$ in the trees of $\T$. Let us denote $Z$ this set. The set $S=S'\setminus Z \cup X$ is then a feedback vertex set of $G$ with $|S|\leq |S'|$ as wanted.

\item[Second case] there is a unique vertex $x\in X\setminus S'$. Then either $S'$ is a feedback vertex set of $G$ (and so we can take $S=S'$) or there is a cycle in $G-S'$, in which case we consider the following two subcases.

\begin{itemize}
    \item If $d_{T}(x)\geq 2$, let $T'\in \T$ with $T'\neq T$ and $d_{T'}(x)\geq 2$ (it exists as $T$ is redundant). Then there is a cycle in $G[T'\cup \{x\}]$, which is hit by the feedback vertex set $S'$ as it is a cycle in $G'$ too. So there is a vertex $v\in T'\cap S'$, and taking $S=\left(S'\setminus \{v\}\right) \cup \{x\}$ is a feedback vertex set of $G$ with $|S|\leq |S'|$.

    \item Otherwise, the graph $G[T\cup \{ x \}]$ has no cycle. The cycle $C$ of $G-S'$ necessarily uses some vertices of $T$ and we have $N_M(T)\setminus S'=\{x\}$ and $G[T\cup \{ x \}]$ without cycle. This configuration is not possible if $T$ is a connected component of $G-M$, so we are in the case where there exists $r$ a common parent for the roots of the trees in $\T$. The cycle $C$ necessarily contains a path between $x$ and $r$ with inner vertices in $T$. Note that all the $|X|+1\ge 2$ trees $T'\in \T\setminus \{T\}$ allow such path between $x$ and $r$. So $S'$ contains at least one vertex in some $T'\in \T\setminus \{T\}$ and by taking one such vertex $v$ we can set $S=(S'\setminus \{v\})\cup \{x\}$ and reach the same conclusion as above.
\end{itemize}
\item[Third case] in the remaining case, where $X\subseteq S'$, we can take $S=S'$.
\end{description}
Hence, the input and output instances are indeed equivalent.
\end{proof}

Observe that \autoref{algo:A} requires a routine for checking if rule~\ruleref{rl:same_neighborhood_prime} applies, before applying it (with~\autoref{lm:applyKR4}).

\begin{lemma}\label{lemma:KR4poly}
Given an instance $(G,k,M,\mH)$, deciding if Rule~\ruleref{rl:same_neighborhood_prime} can be applied, and finding $X$ and $\mT$ if it is the case, can be done in polynomial time.
\end{lemma}
\begin{proof} 
We first compute for each connected component $T$ of $G-M$ their neighborhood in $M$ (i.e., $N_M(T)$). If there is a set $X$ such that the set $\T_X$ of connected components $T$ in $G_M$ such that $N_M(T)=X$ is large enough, that is $|\T_X|\geq |X|+2$, we are done. Finding a family $\T$ for the second version of the rule, where we do not consider connected components of $G-M$ anymore but trees under a common parent $r$, can be done in a similar way by first testing any $r\in V(G-M)$ as the potential common parent. These operations can be performed in polynomial time.
\end{proof}

\subsection{Properties of the kernelized instances}\label{sec:propertiesKernelized}
The goal of this section is to prove that for an instance $(G,k,M,\mH)$ for which the kernelization rules do not apply anymore (meaning that we reach line~\ref{A:partition} in \autoref{algo:A}), the size of a subtree $T$ of $G-M$ is strongly related to $\db(T)$.  (Recall that $\db(T)=\max(d_M(T), \bom(T))$.)

Remember that because the considered graph class $\G$ is nice, it has bounded tree neighborhood complexity for some parameters $\alpha,f_1,f_2$. This implies the easy following lemma.
\begin{lemma}\label{cor:neightrivial}
For every $K_{r,r}$-free $G\in \G$, every set $A\subseteq V(G)$, and every family $\mT$ of disjoint non-adjacent subtrees of $G-A$, and every integer $x$, if $|\mT|\geq xf_1(r)|A|^\alpha$ then there exists $X\subseteq A$ such that at least $x$ subtrees $T\in \mT$ satisfy $N_A(T)=X$. Moreover, suppose that for every $T\in \mT$ we have $d_A(T)\leq p$ and $|T|\leq m$, then if $|\mT|\geq xf_2(r,p,m)|A|$, there exists $X\subseteq A$ such that at least $x$ subtrees $T\in \mT$ satisfy $N_A(T)=X$.
\end{lemma}
\begin{proof}
    The results are obtained from~\autoref{def:NCtrees} by using the pigeonhole principle.
\end{proof}

Recall that given a subtree $T$ of $G-M$, $\pom (T)$ denotes the set $\{v\in V(T),~N(v) \not\subseteq M\cup V(T) \}$, and $\bom (T)$ denotes the size of this set.

We are now ready to bound the degree of the subtrees of $G-M$:

\begin{lemma}\label{lem:degree-T}
Consider an instance $(G,k,M,\mH)$ of \AFVS such that neither Rule~\ruleref{rl:deg1} nor \ruleref{rl:same_neighborhood_prime} applies.
For any subtree $T$ of $G- M$ and any vertex $v$ of $T$,  we have
\[
d_{T}(v)\le \O(\max(\bom(T),f_1(r)d_M(T)^{\alpha+1})).
\]
\end{lemma}

\begin{proof}
Remember that we chose an arbitrary root for $T$, and $T_v$ denotes the subtree of $T$ with root~$v$. For bounding the number of neighbors of $v$ in $T$, it suffice to bound the number of children of $v$. For this, we first use that the number of children $u$ of $v$ such that $V(T_u) \cap \pom(T)\neq \emptyset$ is bounded by $\bom(T)$. Let $N$ be the set of remaining children, meaning the set of children $u$ such that $T_u\cap \pom(T) = \emptyset$.
Applying \autoref{cor:neightrivial}, with the family $(T_u)_{u\in N}$, the vertex set $N_M(T)$ and $x=d_M(T)+2$, if $|N| > (d_M(T)+2)f_1(r)d_M(T)^\alpha $ then there would be a set $X\subseteq N_M(T)$ such that $d_M(T)+2$ vertices $u\in N$ satisfy $d_M(T_u)=X$. Moreover $X\neq \emptyset$ as otherwise Rule~\ruleref{rl:deg1} would have applied. But then Rule~\ruleref{rl:same_neighborhood_prime} would apply.
Hence $|N|=\O(f_1(r) d_M(T)^{\alpha+1})$.
\end{proof}

We now show that the previous rules allow to bound the size of certain types of trees that we define now.

\begin{definition}
A subtree $T$ of $G-M$ is \emph{weakly connected} to $M$ if $G[T\cup \{u\}]$ is acyclic for every $u$ in $M$ (i.e. $d_T(u)\leq 1$ for all $u\in M$).
A subtree $T$ of $G- M$, rooted at a vertex $v$, is \emph{sharp w.r.t. $M$} if $T$ is not weakly connected to $M$ but for every child $u$ of $v$, $T_u$ is weakly connected to $M$.
\end{definition}

\begin{lemma}\label{lem:weak-trees}
Consider an instance $(G,k,M,\mH)$ where none of Rules \ruleref{rl:deg1} and \ruleref{rl:deg2} applies. For any subtree $T$ in $G- M$ that is weakly connected to $M$, we have that $|T|\le 16 \db(T)$.
\end{lemma}
\begin{proof}
Indeed, as Rules \ruleref{rl:deg1} and \ruleref{rl:deg2} do not apply, each vertex $v$ of $T$ belongs to (at least) one of the following types:

\begin{enumerate}
    \item $v\in\pom(T)$,
    \item $d_M(v)\ge 1$,
    \item $d_T(v)\ge 3$, and
    \item $d_T(v)= 2$.
\end{enumerate}

Observe that the leaves of $T$ are either of the first type or the second. Let $l$ be the number of leaves of $T$, there are $\bom(T)$ vertices of the first type, and at most $d_M(T)$ vertices of the second type, so $l\leq \bom(T)+d_M(T)$. Moreover the number of vertices of the third type is at most $l$. Let denote by $Z$ the vertices of the fourth type. 
Observe that the connected components of $G[Z]$ are paths of size at most $3$ (as otherwise \ruleref{rl:deg2} would apply). Replacing such a connected component by an edge between the neighbors of the endpoints of the path would result in a tree with vertices $T\setminus Z$ and whose number of edges (which is at most $|T\setminus Z|-1)$ is an upper bound on the number of connected components of $G[Z]$. So $|Z|\leq 3|T\setminus Z|\leq 6l$, and finally $T$ has less than $8l=8d_M(T)+8\bom(T) \leq 16\db(T)$ vertices in total.
\end{proof}

We now consider sharp subtrees. Notice that in such a tree, $v$ has bounded degree (by \autoref{lem:degree-T}), and the subtrees below $v$ have bounded size (by \autoref{lem:weak-trees}). This leads to the following corollary.

\begin{corollary}\label{cor:sharpsize}
Consider an instance $(G,k,M,\mH)$ where none of Rules \ruleref{rl:deg1}, \ruleref{rl:deg2} and \ruleref{rl:same_neighborhood_prime} applies. Given a sharp subtree $T$ of $G-M$, we have $|T| = O\left(f_1(r)d_M(T)^\alpha\db(T)^2\right)$. 
\end{corollary}
\begin{proof}
    The previous lemmas give 
    \begin{align*}
        |T|&=O\left(\db(T)\max(\bom(T),f_1(r)d_M(T)^{\alpha+1})\right)\\
        &=O\left(f_1(r)d_M(T)^\alpha\db(T)^2\right).
    \end{align*}
\end{proof}

\subsection{Kernelization when a big tree is found.}\label{sec:bigT}
When reaching line~\ref{A:partition} of \autoref{algo:A}, we call the routine \algopart which tries to build a special partition of $G-M$. As we will see later, one output of this procedure is a failure (called Case 1) where a ``big'' (whose size is too large with respect to $d_M(T)$ and $\bom(T)$) tree $T$ is found in $G-M$. 
In this section, we explain how we can get rid of such a big tree.

\begin{enumerate}[label=(\KR{{\arabic*}})]
\setcounter{enumi}{4}
\item \label{rl:neighborhood} Consider an instance $(G, k, M, \mH)$, with a subtree $T$ of $G-M$ which contains $d_M(T)+\bom(T)$ vertex disjoint paths of length at least $1$ and whose endpoints are all adjacent to some vertex $u\in M$. Then output $\left(G-u, k-1, M\setminus\{u\}, \mH-\{u\}\right)$
if $k \ge 1$ and a trivial \textsc{No}-instance otherwise.\footnote{Recall that for a set $V\subseteq V(G)$, $\mH-V$ is defined (at the beginning of \autoref{sssec:kernelization}) as the set of those members of $\mH$ not intersected by $V$, i.e., $\{H \in \mH \mid H\cap V = \emptyset\}$.}
\end{enumerate}

\begin{lemma}
    The rule~\ruleref{rl:neighborhood} is safe.
\end{lemma}
\begin{proof}
Again, as the graph $G'$ obtained after applying the rule is an induced subgraph of $G$, the rule preserves the property of being a $K_{r,r}$-free graph in the considered graph class. For the safeness it is sufficient to prove that if $G$ has a minimum feedback vertex set of size at most $k$, then it has a solution $S'$ of size at most $k$ containing $u$. Let $S$ be a minimum feedback vertex set, and suppose it does not contain $u$. Then $S$ has at least a vertex in each of the $d_M(T)+\bom(T)$ considered disjoints paths in $T$. So the set $S'=(S\setminus T)\cup N_M(T) \cup \pom(T)$ satisfies $|S'|\leq |S|\leq k$ and contains $u$. Moreover this set is a feedback vertex set of $G$ : a cycle in $G-S'$ would contain vertices of $T$ as $S$ is a feedback vertex set of $G$, but then, as $T$ is a tree, such a cycle would need to intersect the set $N_G(T)=N_M(T)\cup \pom(T)\subseteq S'$, which is a contradiction.
\end{proof}

\begin{lemma}\label{lem:subtree-size}
There is a multivariate polynomial $p_1$ with $p_1(x,y) = \O(f_1(x) y^{6+\alpha})$ such that for every instance $(G,k,M,\mH)$ of \AFVS where none of the rules \ruleref{rl:deg1}, \ruleref{rl:deg2}, and \ruleref{rl:same_neighborhood_prime} applies, and for every subtree $T$ of $G-M$ such that the rule~\ruleref{rl:neighborhood} does not apply, we have $|T| \le p_1(r,\db(T))$.
\end{lemma}
\begin{proof}
Note that a sharp subtree of $T$ contains a path whose endpoints are distinct and adjacent to a same vertex $u\in M$. So as the rule~\ruleref{rl:neighborhood} does not apply, we have that $T$ contains at most $d_M(T)(d_M(T)+\bom(T))$ disjoint sharp subtrees.

Let $T^1=T$ (and remember that the trees of $G-M$ are rooted).
Then, iteratively let $v_i$ be a lowest vertex of $T^i$ such that $T^i_{v_i}$ is sharp, and let $T^{i+1} = T^i \setminus T^i_{v_i}$.
This process stops when the remaining tree is weakly connected to $M$ or empty. This leads to a partition of $T$ into $c+1$ disjoint subtrees, where the first $c$ are sharp subtrees, and the last one is weakly connected to $M$, with $c\le d_M(T)(d_M(T)+\bom(T))=\O(\db(T)^2)$.  Using the bound on the size of the weakly connected and sharp trees from \autoref{lem:weak-trees} and \autoref{cor:sharpsize}, we get:

\begin{align*}
    |T| &= |T^{c+1}| + \sum_{1\le i\le c} |T^i_{v_i}|\\
           &= \O\left(\db(T^{c+1})+\sum_{1\le i\le c} f_1(r) d_M(T^i_{v_i})^\alpha\db(T^i_{v_i})^2\right) ~ \text{by \autoref{lem:weak-trees} and \autoref{cor:sharpsize}}\\
           &= \O\left(\db(T)+f_1(r) d_M(T)^\alpha\sum_{1\le i\le c} \db(T^i_{v_i})^2\right)\\
           &= \O\left(\db(T)+f_1(r) d_M(T)^\alpha \left(\sum_{1\le i\le c} \db(T^i_{v_i})\right)^2 \right).
\end{align*}
Note that each vertex of $\pom(T)$ belongs to exactly one of these subtrees, and each subtree $T^i_{v_i}$ creates exactly one vertex in $\pom(T^{i+1})$, and so will contribute to at most one of the $\bom(T^j_{v_j})$ for $j>i$. Hence, $\sum_i \bom(T^i_{v_i}) \le \bom(T) + c =\O(\db(T)^2)$. 
Moreover observe that $\sum_i d_M(T^i_{v_i}) \le c d_M(T) = \O(\db(T)^3)$ so $\sum_i \db(T^i_{v_i})=\O(\db(T)^3)$.

We then obtain $|T|=\O(f_1(r) \db(T)^{6+\alpha})$.
\end{proof}

\begin{corollary}\label{lm:boundSzT}
    Given an instance $(G,k, M, \mH)$ where none of Rules \ruleref{rl:deg1}, \ruleref{rl:deg2}, and \ruleref{rl:same_neighborhood_prime} applies, and given a subtree $T$ of $G-M$ with $|T|> p_1(r,\db(T))$, Rule~\ruleref{rl:neighborhood} can be applied in polynomial time on this tree.
\end{corollary}
\begin{proof}
    The fact that Rule~\ruleref{rl:neighborhood} applies is the contrapositive of \autoref{lem:subtree-size}, and the $d_M(T)+\bom(T)$ paths needed to apply the rule can be found using the partition used in its proof, which can be computed in polynomial time.
\end{proof}

\subsection{Branching when there is a \texorpdfstring{$\KttTilde$}{K~tt}} \label{sec:ktttilde}

When reaching line~\ref{A:partition} of \autoref{algo:A}, we call the method \algopart which tries to build a special partition of $G-M$. As we will see later, one of the possible outputs of this procedure is a failure (called Case 2) where a certain dense structure called $\KttTilde$ is found (for $t=2\d(r)$).
In this section, we explain how we handle such a $\KttTilde$.
Informally, a $\KttTilde$ in an instance $(G,k,M \mH)$ is a special type of minor model of $K_{t,t}$ where the vertices of one side of the bipartition of $K_{t,t}$ correspond to single vertices of $M$ and the vertices of the other side of the bipartition correspond to small trees of $G-M$. It is formally defined as follows.

\begin{definition}
For an integer $t$, a $\KttTilde$ in 
an instance $(G,k,M,\mH)$ of \AFVS is a pair $\left(X, \{T_i\}_{1\leq i \leq t}\right)$ where
\begin{itemize}
    \item $X\subseteq M$ has size $t$, 
    \item  $\{T_i\}_{1\leq i \leq t}$ a family of $t$ disjoint non-adjacent subtrees of $G - M$ such that for all $1\leq i \leq t$, $X\subseteq N_M(T_i)$, and
    \item for every $i \in \intv{1}{t}$, $|T_i| \le p_2(r,t)$ with $p_2(r,t)=p_1(r,2t)$, where $p_1$ is the polynomial defined in \autoref{lem:subtree-size}.
\end{itemize}
\end{definition}

\begin{remark}\label{rmk:Ktttilde}
  Similarly to a $K_{t,t}$ subgraph, given $(X, \{T_i\}_{1\leq i \leq t})$ a $\KttTilde$, any subgraph of $G$ induced by two vertices of $X$ and two trees of $\{T_i\}_{1\leq i \leq t}$ contains a cycle. So a feedback vertex set of $G$ either contains at least $t-1$ vertices of $X$, or it intersects all but at most one of the trees of $\{T_i\}_{1\leq i \leq t}$.
\end{remark}
  
This new branching rule is a variation of Rule~\ruleref{rl:Krr} to deal with $\KttTilde$'s instead of $K_{t,t}$ subgraphs.
The difference between the treatment of $K_{t,t}$ and $\KttTilde$ is that, when we are handling the case where a solution must intersect each tree of $\{T_i\}_{1\leq i \leq t}$, we cannot branch in subexponential time to guess which vertex of each $T_i$ is picked, and rather perform the following ``virtual branching'' trick that generalizes that of \cite{lokSODA22} corresponding to the special case where $T_i$ are edges. In this branch, we add all these $T_i$'s in $M$, and add the vertex sets of each tree to $\mH$. The size of $\mH$ (and thus the growth of $M$) will be bounded by observing that since the sets in $\mH$ are disjoint, $|\mH|$ is a lower bound on the solution size. 

We are now ready to define our last branching rule:

\begin{enumerate}[label=(\BR{{\arabic*}})]
\setcounter{enumi}{1}
\item \label{rl:KrrTilde} Given $(G, k, M, \mH)$ an instance of \AFVS{} and $(X, \{T_i\}_i)$ a $\KttTilde$, we 
branch on the following (at most $2t$) instances:
\begin{itemize}
    \item if $k \ge t-1$, then for each $v\in X$, denoting $X_{\overline{v}}=X\setminus \{v\}$, output the instance $((G-X_{\overline{v}}), k-(t-1), M\setminus X_{\overline{v}}, \mH-X_{\overline{v}})$.
    \item if $k \geq |\mH|+t-1$, then for each $1\leq i \leq t$, denoting $R_i=\bigcup_{j\neq i}T_j$ and $\mT_{R_i}=\{T_j,~j\neq i\}$, output the instance $\left(G, k, M\cup R_i, \mH\cup \mT_{R_i}\right)$.
    \item if none of the above applies, output a trivial \textsc{No}-instance.
\end{itemize}
\end{enumerate}

\begin{lemma}
    Given $(G, k, M, \mH)$ an instance of \AFVS{} and $(X, \{T_i\}_i)$ a $\KttTilde$, the rule~\ruleref{rl:KrrTilde} produces a set of instances of \AFVS{} such that $(G, k, M, \mH)$ is a positive instance if and only if one of the output instances is positive. Moreover, the rule can be performed in polynomial time.
\end{lemma}
\begin{proof}
    The obtained graph is an induced subgraph of the original graph, so it belongs to $\mathcal{G}$ and is $K_{r,r}$-free. Also, recall that $\mH$ consists of disjoint subsets of $M$. In particular, in the second item of \ruleref{rl:KrrTilde} these sets are disjoint from those in $\mT_{R_i}$ as for every $i\in \intv{1}{t}$, $T_i$ is a subtree in $G-M$. Therefore $\mH\cup \mT_{R_i}$ is a collection of disjoint subsets of $M\cup R_i$. Hence the rule~\ruleref{rl:KrrTilde} indeed outputs instances of \AFVS.
    The running time bound is straightforward as it is not required to find the $\KttTilde$. 
    
    The fact that the input and output instances are equivalent comes from \autoref{rmk:Ktttilde}.
    More formally, suppose that $S$ is a solution to $(G,k,M,\mH)$. If $|S \cap X| \ge t-1$, 
    then there exists $v \in X$ such that $S \supseteq X_{\overline{v}}$, implying that 
    $((G-X_{\overline{v}}), k-(t-1), M\setminus X_{\overline{v}}, \mH-X_{\overline{v}})$ is a \textsc{Yes}-instance. Otherwise ($|S \cap X| \le t-2$), there exists $u,v$ two vertices in $X \setminus S$. As for any $i$, $G[T_i]$ is a tree and $X \subseteq N_M(T_i)$, there cannot be $i \neq j$ such that $S \cap T_i = S \cap T_j = \emptyset$, as otherwise $G[\{u,v\} \cup T_i \cup T_j]$ would contain a cycle not hit by $S$. This implies that there exists $i$ such that all trees of $\mT_{R_i}$ are hit by $S$.
    As $S$ must also hit any tree in $\mH$ (which are disjoint), it implies that $k \geq |\mH|+t-1$, and thus the rule generates in particular $\left(G, k, M\cup R_i, \mH\cup \mT_{R_i}\right)$, which is a \textsc{Yes}-instance.
    The reverse direction is immediate.
\end{proof}

\subsection{Attempting to build a \texorpdfstring{$t$}{t}-uniform partition of  \texorpdfstring{$G-M$}{G-M}}
\label{sssec:partition-algo}
In this section we define the routine \algopart that is called line~\ref{A:partition} when \ruleref{rl:deg1}, \ruleref{rl:deg2}, \ruleref{rl:big}, and \ruleref{rl:same_neighborhood_prime} do not apply.
The algorithm will
either fail (Case 1) and find a large enough subtree $T$ that allows to apply \ruleref{rl:neighborhood} of \autoref{sec:bigT},
fail (Case 2) and find a $\KttTilde$ that allows to apply \ruleref{rl:KrrTilde} of \autoref{sec:ktttilde},
or find a certain partition of $G-M$ that we define now.

\begin{definition}\label{def:unif}
Let $(G,k,M,\mH)$ be an instance of \AFVS and $t$ be a positive integer. Let $F \subseteq V(G-M)$ be a downward-closed subforest of $G-M$.
A \emph{$t$-uniform partition of $F$} is family $\mT$ of subtrees of $F$ such that:
    \begin{enumerate}
    \item \label{item:unif1} each vertex of $F$ is in exactly one tree of $\mT$,
    \item \label{item:unif2} for every $T\in \mT$ have $d_M(T)\le 2t$ and $|T|\leq p_1(r,\db(T))$ (where $p_1$ is the polynomial function defined in \autoref{lem:subtree-size}), and
    \item \label{item:unif3} if a root $r$ of a tree of $\mT^-$ has a parent $v$ in $G-M$, then $v$ is the root of a tree in $\mT^+$,
\end{enumerate}
where $\mT^-$ denote the subset of those $T\in \mT$ such that $d_M(T)<t$ and $\mT^+$ the others.
\end{definition}
\begin{remark}
The definition above implies that there are no edges between the trees of $\mT^-$.
\end{remark}
\begin{figure}
    \centering
    \includegraphics{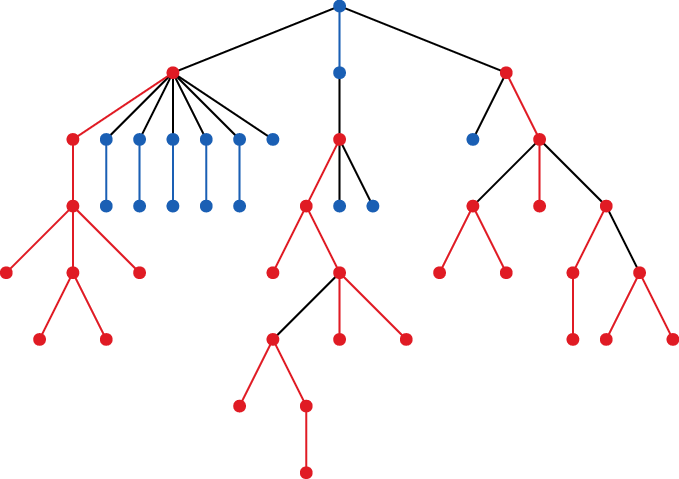}
    \caption{Representation of a $t$-uniform partition of a subtree $T$ of $G-M$. Here are represented only vertices of $F$ (and not vertices of $M$). Vertices of the trees in $\mT^+$ (respectively $\mT^-$) are represented in red (respectively blue), so as the edges between two vertices in the same tree of $\mT^+$ (respectively $\mT^-$). Edges between distinct trees of $\mT$ are represented in black. }
    \label{fig:uniform}
\end{figure}

One case of failure of \algopart when trying to build a $t$-uniform partition is when $\mT^+$ becomes too large. We show in the next lemma that in this case we can find a $\KttTilde$.
\begin{lemma}\label{lm:bigreduce}
   Consider an instance $(G, k, M, \mH)$ where none of Rules \ruleref{rl:deg1}, \ruleref{rl:deg2}, and \ruleref{rl:same_neighborhood_prime} applies,
 a downward-closed subforest $F \subseteq V(G-M)$ and a $t$-uniform partition $\mT$ of $F$ with $|\mT^+| \ge p_3(r,t)|M|$ for $p_3(r,t)=4tf_2(r,2t,p_2(r,t))$. 
 Then, we can find a $\KttTilde$ in polynomial time.
\end{lemma}
\begin{proof}
    Consider $F_{\mT^+}$ the forest obtained from $G-M$ by contracting each subtree in $\mT^+$ to a single vertex, and removing the vertices of the trees in $\mT^-$. Denoting $V_s$ the set of vertices of $F_{\mT^+}$ with degree at most $2$ in $F_{\mT^+}$. We have $|V_s|\geq |\mT^+|/2$ as the number of leaves in a tree is an upper bound on the number of vertices of degree at least $3$.
    Then using a 2-coloration of the forest we partition the vertices of $F_{\mT^+}$ in two independent sets, $V_1$ and $V_2$. 
    Without loss of generality we suppose that $|V_s\cap V_1|\geq |V_s\cap V_2|$. So denoting $\mT'$ the set of trees in $\mT^+$ associated to the vertices in $V_s\cap V_1$, we have $|\mT'|\geq |V_s|/2\geq |\mT^+|/4$. Now for $T\in \mT'$ we want to bound the size of its border. 
    Firstly,  for the trees in $\mT^+$ adjacent to $T$, there are at most two of them by construction of $\mT'$.
    Secondly,  by \autoref{item:unif3}  of \autoref{def:unif} the only vertex in $T$ which can be adjacent to a tree $T'\in \mT^-$ is the root of $T$. 
    So overall, the border of $T$ contains at most $3$ vertices, and so $\bom(T)\leq 3$. Moreover, as $V_1$ is an independent set, there is no edge between two trees of $\mT'$. By \autoref{def:unif}, any tree of $T\in \mT'$ has size at most $p_1(r,\db(T))$, but we have $\db(T)=\max(d_M(T), \bom(T))=d_M(T)\leq 2t$ so $|T|\leq p_1(r,2t)=p_2(r,t)$. Let us now apply the second result of \autoref{cor:neightrivial} with $x=t$, $A=M$, $p=2t$ and $m=p_2(r,t)$.
    As $|\mT'|\geq |\mT^+|/4 \ge tf_2(r,2t,p_2(r,t))|M|$, there exists a subset $X\subseteq M$ and a subset $\mT'_X\subseteq\mT'$ such that $|\mT'_X|=t$, and for all $T\in \mT'_X$, we have $N_M(T)=X$. Observe that as $\mT'\subseteq \mT^+$, we have $|X|\geq t$, so we can take  any subset $X_0$ of $X$ such that $|X_0|=t$ and then $(X_0, \mT'_X)$ is a $\KttTilde$, as wanted.
    
    Regarding the running time, note that the set $\mT'$ can be computed in polynomial time, and finding the subset $X$ and $\mT'_X$ whose existence is proved by \autoref{cor:neightrivial} can again be done in polynomial time by listing the neighborhoods in $M$ of the trees in $\mT'$. 
\end{proof}

We are now finally ready to describe the aforementioned \algopart{}, which is done with the following lemma.
\begin{lemma}[Partitionning lemma]\label{lem:construct-red-trees}
Let $t$ be an integer. There exists a polynomial time algorithm that,
given an instance $(G, k, M, \mH)$ of \AFVS{} where none of Rules \ruleref{rl:deg1}, \ruleref{rl:deg2}, \ruleref{rl:same_neighborhood_prime} applies and such that $d_M(v)< t$ for every $v\in V(G-M)$, returns one of the following:
\begin{enumerate}[label=Case~\arabic*]
    \item \label{part:case1} a subtree $T$ of the forest $G-M$ with $|T|>p_1(r,\db(T))$
    \item \label{part:case2} a $\KttTilde$, or 
    \item \label{part:case3} a $t$-uniform partition $\mT$ of the forest $G-M$ with $|\mT^+|<p_3(r,s,t)|M|$, where $p_3$ is defined in \autoref{lm:bigreduce}.
\end{enumerate}
\end{lemma}

Observe that in \ref{part:case1} and \ref{part:case2} above, the rules \ruleref{rl:neighborhood} and \ruleref{rl:KrrTilde} respectively apply, which will allow us to make progress even if we cannot find a partition.

\begin{proof}
In this proof we set $p=p_3(r,s,t)|M|$.
Let $h$ be the maximum depth of a leaf in $G-M$ and let $L_0$ be the set of roots of the connected components of $G-M$. For every $i\in \intv{1}{h}$, we define $L_i$ as the set of vertices of depth $i$ in the trees of the forest $G-M$ and $F_i$ as the forest induced by the set of vertices with depth at least $i$, i.e., $F_i=G\left [\bigcup_{j=i}^h L_j \right ]$. Observe that $F_0=G-M$.

The algorithm we describe here considers vertices by decreasing depth. First it construct a (trivial) partition of the forest $F_h$ induced by the vertices of maximum depth.
Observe that two vertices of the same depth cannot share an edge, so $F_h$ an edgeless graph. In order to construct the partition we then have no choice but to take each connected component (containing only one vertex) as a tree in $\mT$, and more precisely in $\mT_h^-$ as for $v\in V(F_h)$ we have $d_M(v) < t$ by hypothesis, and so $\mT_h^+ = \emptyset$.
The collection of trees $\mT_h$ is then a $t$-uniform partition as there is no edge between two trees of $\mT_h$, and each tree have degree in $M$ at most $2t$ as wanted (even $t$, as noted above).

The second stage of the algorithm consists in repeatedly applying a routine to try to extend a $t$-uniform partition of the vertices at depth more than $i$ to a $t$-uniform partition of the vertices at depth more than $i-1$, for decreasing $i$. Possibly this extension will fail, in which case we show that we can return a large subtree of $G-M$ or a $\KttTilde$. If the successive extensions succeed up to partitioning $V(G-M)$, we show that the result is a $t$-uniform partition. The routine is the following.

\begin{enumerate}[label=(\arabic*)]
    \item \label{l:input} Input: a $t$-uniform partition $\mT_i$ of $F_i$ such that $|\mT_i^+| < p$.
    \item \label{l:output} Output: a $t$-uniform partition $\mT_{i-1}$ of $F_{i-1}$ with $|\mT_{i-1}^+| < p$, or a $\KttTilde$, or a subtree of $G-M$ with more than $p_1(r,\db(T))$ vertices. The fact that the routine described here indeed corresponds to this specification is given in \autoref{claim:part}.
    \item Initialization: we start with $\mT_{i-1}^+ = \mT_{i}^+$ and $\mT_{i-1}^- = \mT_{i}^-$. The collection $\mT_{i-1}$ is always defined as the union $\mT_{i-1}^- \cup \mT_{i-1}^+$.
    \item \label{e:for} For every $v\in L_{i-1}$:
    \begin{enumerate}[label=(\ref*{e:for}.{\arabic*})]
        \item Let $v_1,\dots, v_l$ denote the children of $v$, which all belong to $L_{i}$. Observe that by definition of a $t$-uniform partition of $F_{i}$, each such vertex $v_j$ is the root of some tree $T(v_j)$ in $\mT_i$. Informally, the goal will be to try to group some of the trees of the form $T(v_j) \in \mT_i^-$ to construct a tree to add in $\mT_{i-1}^+$.
    \item Let $Y^-(v) = \{j \mid T(v_j) \in \mT_i^-\}$.
    \item If $d_M({v}\cup \bigcup_{j \in Y^-(v)}T(v_j)) < t$, 
set $X(v)= Y^-(v)$. Otherwise, define $X(v)$ as an inclusion-wise minimal subset of $Y^-(v)$ such that $d_M\left(\{v\}\cup \bigcup_{j \in X(v)}T(v_j)\right) \ge t$.
    \item Let 
$T(v) = \{v\}\cup \bigcup_{j \in X(v)}T(v_j)$.
    \item \label{l:tree} If $|T(v)| > p_1(r,\db(T))$ output the tree $T(v)$ and stop.
    %
    \item If $d_M(T(v)) < t$, remove from $\mT_{i-1}^-$ all the $T(v_j)$ for $j \in X(v)$, and add $T(v)$ to $\mT_{i-1}^-$. Vertex $v$ has been dealt with: continue in line~\ref{e:for} to the next choice of a vertex in $L_{i-1}$.
    \item Otherwise ($d_M(T(v)) \ge  t$), remove from $\mT_{i-1}^-$ all the $T(v_j)$ for $j \in X(v)$, and add $T(v)$ to $\mT_{i-1}^+$.
    \item If $|\mT_{i-1}^+|<p$, continue in line~\ref{e:for} to the next choice of vertex of $L_{i-1}$.
    \item \label{l:ktt} Otherwise, $|\mT_{i-1}^+|=p$. In this case we apply \autoref{lm:bigreduce} to the instance $(G, k, M, \mH)$ with partition $\mT_{i-1}$ of the forest $F=\bigcup_{T \in \mT_{i-1}} T$.
    As we will show, $\mT_{i-1}$ is a $t$-uniform partition of $F$ so this call results in a $\KttTilde$. Output the obtained $\KttTilde$ and stop.
      \end{enumerate}
    \item \label{l:part} Return $\mT_{i-1}$.
\end{enumerate}

The correctness of the above routine is given by the following claim.

\begin{claim}\label{claim:part}
The routine above is correct, i.e., given an input as described in line~\ref{l:input} above, the routine returns an output as specified in line~\ref{l:output}.
Moreover, the computation can be done in polynomial time.
\end{claim}

\begin{proof}
Let us consider as input $\mT_i$ a $t$-uniform partition of $F_i$ with $i\geq 1$ such that $|\mT_i^+| < p$.
We consider an execution of the routine on an input as described in line~\ref{l:input} and that considers vertices $v \in L_{i-1}$ in an arbitrary order.

There are 3 places where the routine returns a value.
If this happens at line~\ref{l:tree}, the returned value is clearly a subtree of $G-M$ with more than $p_1(r,\db(T))$ vertices. So in the following we may assume that we are not in this case and so for every vertex $v\in L_{i-1}$ considered so far, $|T(v)| \le p_1(r,\db(T))$.

If we return at line~\ref{l:ktt} then the output is a $\KttTilde$ according to \autoref{lm:bigreduce}, supposing we gave an appropriate input to the algorithm of that lemma.
By construction $F$ is downwards-closed.
To show that $\mT_{i-1}$ is a $t$-uniform partition of $F$, let us consider the 3 items in \autoref{def:unif}.
\autoref{item:unif1} holds by construction. For every tree of $\mT_{i-1}$ of the form $T(v)$, either $v\in V(F_i)$ and so the bounds of \autoref{item:unif2} hold since $\mT_i$ is a uniform $t$-partition, or $v\in L_{i-1}$ and as noted above we assume $|T(v)| \le p_1(r,\db(T))$. It remains to show that for such a vertex $v$ we have $d_M(T(v)) \leq 2t$. To do so, notice that when $d_M(T(v)) \ge t$, 
then $X(v) \neq \emptyset$ as otherwise we would get $d_M(T(v))=d_M(v) < t$, a contradiction. We can thus consider any $j \in X(v)$, and observe that $d_M(T(v)) \le d_M\left(\{v\}\cup \bigcup_{k \in X(v)\setminus \{j\}} T(v_k)\right) + d_M\left(T(v_j)\right) < 2t$ by minimality of $X(v)$, and as $T(v_j) \in \mT_{i-1}^-$.
To show \autoref{item:unif3}, we will use the following invariant.
\begin{itemize}
    \item At any time during the execution of the routine, if a root $r$ of a tree in $\mT_{i-1}^-$ has a parent $u$ in $G-M$, then either $r\in L_{i-1}$ and $u  \in L_{i-2}$, or $r\in L_{i}$, $u \in L_{i-1}$, and $u$ either has not been considered so far or is the root of a tree in $\mT_{i-1}^+$.
\end{itemize}
The proof of the invariant is the following. Let $r$ be the root of a tree $T \in \mT_{i-1}^-$ which has a parent $u$ in $G-M$. If $r \in L_{j}$ for $j >i$, by hypothesis $\mT_i$ is a $t$-uniform partition of $F_i$, then $u$ is the root of a tree $T \in \mT_i^+$, and $T \in \mT_{i-1}^+$ as we never remove trees from $\mT_i^+$.
If $r \in L_{i-1}$, then $u \in L_{i-2}$ and we are done.
Finally, assume that $r \in L_{i}$. If $u$ has not been considered so far then we are done. Otherwise, as $u$ has been considered and $r$ is still the root of a tree $T \in \mT_{i-1}^-$, it implies that $X(u)\neq Y^-(u)$, and thus that $T(u)$ was added to $\mT_{i-1}^+$, satisfying the required condition. This concludes the proof of the invariant.

We can now show that \autoref{item:unif3} of \autoref{def:unif} holds for $\mT_{i-1}$. By the first item of the invariant, if a root $r$ of a tree in $\mT_{i-1}^-$ has a parent $u$ in $G-M$, then $u  \in L_{i-2}$ (and thus does not belong to $F$), or $u \in L_{i-1}$ and $u$ has not been considered so far (so $u$ does not belong to $F$), or $u$ is the root of a tree in  $\mT_{i-1}^+$. This shows
that $\mT_{i-1}$ is indeed a $t$-uniform partition of $F$ hence the call to the algorithm of \autoref{lm:bigreduce} is valid. This concludes the proof that in line~\ref{l:ktt} the routine returns a $\KttTilde$.

It remains to show that in line~\ref{l:part}, $\mT_{i-1}$ is a $t$-uniform partition of $F_{i-1}$ with $|\mT_{i-1}^+| < p$.
If this line is reached, then in particular the routine did not stop in line~\ref{l:ktt} and so we have $|\mT_{i-1}^+| < p$.
Above we proved that after each iteration of the loop of line~\autoref{e:for}, $\mT_{i-1}$ is a $t$-uniform partition of the forest induced by $V(F_i)$ and the vertices of $L_{i-1}$ considered so far. In line~\ref{l:part} all the vertices of $L_{i-1}$ have been considered, so $\mT_{i-1}$ is a $t$-uniform partition of $F_{i-1}$, as desired.
\cqed
\end{proof}

Now that  \autoref{claim:part} is proved, \autoref{lem:construct-red-trees} directly holds by induction on $i$, using $\mT_h$ for the base case.
\end{proof}

\begin{remark}
Observe that we call \algopart of \autoref{lem:construct-red-trees} at line~\ref{A:partition} of \autoref{algo:A}, all required conditions on the input are fulfilled, as in particular as \ruleref{rl:big} does not apply, we get $d_M(v) \le 2\d(r) \le t$ for any $v \in V(G-M)$.
\end{remark}

\subsection{Final step when \texorpdfstring{$G-M$}{G-M} admits a \texorpdfstring{$t$}{t}-uniform partition with a small number of large-degree parts}\label{ssec:step4}

In this section, we consider line~\ref{A:partOK} of \autoref{algo:A} where we found a $t$-uniform partition $\mT$ as required. To bound $|V(G)|$, the algorithm adds some sets $Z_1(\mT)$ and $Z_2(\mT)$ (as defined below) to $M$, which result in a slightly larger set $\tilde{M}$. It then tries to apply Rule~\ruleref{rl:same_neighborhood_prime} or Rule~\ruleref{rl:neighborhood} on $(G,k,\tilde{M},\mH)$. If none of these rules can be applied, we can prove that $|V(G-M)|$ is small.

\begin{definition}\label{def:Z}
    Given a $t$-uniform partition $\mT$ of $G-M$, we define the set $Z_1(\mT)$ as the roots of the trees in $\mT^+$, and $Z_2(\mT)$ as the set of vertices $v$ of $G - M$ having three edge-disjoint paths $P_1,P_2,P_3$ in $G- M$ linking them to vertices of $Z_1(\mT)$.
\end{definition}

\begin{lemma}\label{lm:special}
    Given an instance $(G, k, M, \mH)$, an integer $t$, and $\mT$ a $t$-uniform partition of $G-M$, we have $|Z_1(\mT)|+|Z_2(\mT)|\leq 2|\mT^+|$, and by denoting $\tilde{M}=M\cup Z_1(\mT)\cup Z_2(\mT)$, we have the following properties:
    \begin{itemize}
        \item Rule~\ruleref{rl:deg1} and Rule~\ruleref{rl:deg2} do not apply on $(G,k,\tilde{M},\mH)$.
        \item A connected component $T$ of $G-\tilde{M}$ satisfies $d_{\tilde{M}}(T)\leq 2t+2$.
    \end{itemize}
\end{lemma}
\begin{proof}
    We denote $Z_1=Z_1(\mT)$ and $Z_2=Z_2(\mT)$. The bound $|Z_1|\leq |\mT^+|$ is trivial.
    To bound $|Z_2|$, let us define the following
    minor $G_Z$ of $G - M$. The graph $G_Z$ is obtained from $G - M$ by iteratively deleting leaves $\ell$ if $\ell\notin Z_1\cup Z_2$, or contracting every edge $zx$ with $z\in Z_1\cup Z_2$ and $x\notin Z_1\cup Z_2$. Clearly $G_Z$ is a forest with vertex set $Z_1\cup Z_2$. Furthermore, for every $v\in Z_2$, the three edge-disjoint paths linking $v$ to $Z_1$ are actually vertex-disjoint and end up on different vertices of $Z_1$, hence $d_{G_Z}(v)\ge 3$. 
    As the number of vertices of degree at least three in a tree is bounded by the number of leaves, we get $|Z_2|\leq |Z_1| \le |\mT^+|$.
    
    The non applicability of Rule~\ruleref{rl:deg1} and Rule~\ruleref{rl:deg2} on $(G,k,\tilde{M},\mH)$ immediately follows from the non applicability of Rule~\ruleref{rl:deg1} and Rule~\ruleref{rl:deg2} on $(G,k,M,\mH)$.
    
    It then remains to bound $d_{\tilde{M}}(T)$ for $T$ a connected component of $G-\tilde{M}$.  
    Let $T$ be a connected component of $G-\tilde{M}$. We have $T \subseteq T'$ for some $T' \in \mT$, as each edge between trees of $\mT$ contains a vertex from $Z_1(\mT)\subseteq \tilde{M}$, and as $\mT$ is a $t$-uniform partition, we get $d_M(T) \le d_M(T') \le 2t$. Moreover, notice that $d_{\tilde{M}}(T) \le d_M(T)+|\{vz \in E(G)\ \mid\ v \in T \text{ and } z\in Z_1\cup Z_2\}|$.
    Thus, assume by contradiction that $d_{\tilde{M}}(T) \ge 2t+3$, and thus there are three (non-necessarily distinct) vertices $v_1,v_2,v_3$ in $T$ and three distinct vertices $z_1, z_2, z_3$ in $Z_1 \cup Z_2$ such that $v_1z_1$, $v_2z_2$, and $v_3z_3$ are edges. Then $T$ contains a vertex $v$ linked by 3 edge-disjoint paths to $Z_1\cup Z_2$. As such a path ending in $Z_2$ could be prolonged towards $Z_1$ while staying edge disjoint from the others, such a vertex $v$ should belong to $Z_2$, a contradiction.
\end{proof}

\begin{lemma}\label{lm:finalR4R5}
    Given an instance $(G, k, \tilde{M}, \mH)$ such that Rules \ruleref{rl:deg1}, \ruleref{rl:deg2}, and \ruleref{rl:same_neighborhood_prime} do not apply, and such that for any $T$ connected component of $G-\tilde{M}$, Rule~\ruleref{rl:neighborhood} does not apply on $T$ and $d_{\tilde{M}}(T)\leq 2t+2$. Then, $|G-\tilde{M}| \le p_4(r,t)|\tilde{M}|$ with $p_4(r,t)=(2t+4)f_2(r,2t+2,p_1(r,2t+2))$.
\end{lemma}
\begin{proof}
Let $\mC$ be the set of connected components of $G-\tilde{M}$. Let $T \in \mC$. As Rule~\ruleref{rl:deg1}, \ruleref{rl:deg2}, and \ruleref{rl:same_neighborhood_prime} do not apply, and Rule~\ruleref{rl:neighborhood} does not apply on $T$, by \autoref{lm:boundSzT} we get $|T| \le p_1(r,\db(T)) \le p_1(r,2t+2)$.  Now,  as Rule~\ruleref{rl:same_neighborhood_prime} does not apply in particular with $\mC$, by the second result of \autoref{cor:neightrivial} with $A=\tilde{M}$, $x=2t+4$, $p=2t+2$ and $m=p_1(r,2t+2)$, we get that if $|\mC| > xf_2(r,p,m)|\tilde{M}|$, then we would have $x$ connected components of $\mC$ having the same neighborhood $X$ in $\tilde{M}$.
Then we would have $X\neq \emptyset$ as otherwise Rule~\ruleref{rl:deg1} would apply. But then as $|X| \le p=2t+2$ and $x \ge |X|+2$, we could apply \ruleref{rl:same_neighborhood_prime}, a contradiction with the hypothesis.
Thus, we get $|G-\tilde{M}| \le (2t+4)f_2(r,2t+2,p_1(r,2t+2))|\tilde{M}|$.
\end{proof}

\begin{lemma}\label{lm:sizeGMinusM}
Let $(G,k,M,\mH)$ be an instance of \AFVS.
If on this input \autoref{algo:A} reaches Line~\ref{A:DP}, 
then $|G-M|=\O(p_3(r,t)p_4(r,t)|M|)$.
\end{lemma}
\begin{proof}
If \autoref{algo:A} reaches line~\ref{A:DP}, it implies that \algopart output a $t$-uniform partition of $G-M$ (where $t=2\d(r)$ and $|\mT^+| \le p_3(r,s,t)|M|$).
Let $\tilde{M}=M \cup Z_1(\T) \cup T_2(\T)$.
By \autoref{lm:special}, Rule~\ruleref{rl:deg1} and Rule~\ruleref{rl:deg1} do not apply on $(G,k,\tilde{M},\mH)$.
Also, since we reached this line it means that we did not apply Rule~\ruleref{rl:same_neighborhood_prime} on Line~\ref{A:KR4} or Rule~\ruleref{rl:neighborhood} on Line~\ref{A:KR5} (which both return the result of a recursive call). By \autoref{lm:special}, $d_{\tilde{M}}(T) \le 2t+2$ for any connected component $T$ of $G-\tilde{M}$, so we can apply \autoref{lm:finalR4R5} to obtain
$|G-\tilde{M}| \le p_4(r,t)|\tilde{M}|$.
As $|G-M|=|G-\tilde{M}|+|Z_1(\T)|+|Z_2(\T)|$, we get
$|G-M|=\O\left(p_4(r,t)(|M|+|Z_1(\T)|+|Z_2(\T)|)\right)$.
By \autoref{lm:special} we have $|Z_1(\T)|+|Z_2(\T)| \le 2|\mT^+|$ and $|\mT^+| \le p_3(r,s,t)|M|$, we get
$|G-M|=\O(p_3(r,t)p_4(r,t)|M|)$.
\end{proof}

Let us now bound the size of $M$ in any call to \autoref{algo:A} by providing invariants on its inputs.
Informally, $|M|$ can be bounded as follows.
Recall that we perform our first call to \autoref{algo:A} on input $(G_i,k_i,M_i,\emptyset)$ where $(G_i,k_i,M_i) \in \I$ and $M_i \subseteq M_0$, with $|M_0| \le 2k_0$.
Then, if we follow any path in the tree of calls whose root is the call on input $(G_i,k_i,M_i,\emptyset)$, we may add some vertices to $M$ because of \ruleref{rl:big} (let denote $A_1$ those vertices) or because of \ruleref{rl:KrrTilde} (we denote $A_2$ those vertices). On a given branch the rule~\ruleref{rl:KrrTilde} add at most $k_0$ trees of size at most $p_2(r,t)$ so $|A_2|\leq k_0p_2(r,t)$. It remains to bound the vertices of $A_1$. Observe that when a vertex $v$ is added to the set $M$ because of \ruleref{rl:big} (and so to $A_2$) we have $d_M(v)\geq 2\d(r)$, and so the number of edges in the graph $G[M]$ increases by at least $2\d(r)$. We may think that as we have at most $\d(r)|M|$ edges in $G[M]$, we obtain $2\d(r)|A_1|\leq \d(r)|M| \leq \d(r)(|M_0|+|A_1|+|A_2|)$ which gives a bound on the size of $A_1$. However this reasoning is flawed as we actually need to consider all the vertices added to $M$, even those deleted because of \ruleref{rl:neighborhood}. But then, if we consider deleted vertices then some vertices which were handled by \ruleref{rl:deg2} may not have degree 2 anymore, and the graph we need to consider possibly does not live in $\G$, and in fact may even contains a $K_{r,r}$ as a subgraph. Therefore a more technical analysis is needed to bound the size of $M$.

\begin{lemma}\label{lm:sizeM}
In a run of \autoref{algo:A} on some input $(G_i,k_i,M_i,\emptyset)$, where $(G_i,k_i,M_i) \in \I$, there may be several calls to \autoref{algo:A}, from the initial one to possible recursive calls.
Any such call is made on an input of the form $(G,k,M,\mH)$ with $|M|=\O(k_0p_2(r,t))$.
\end{lemma}
\begin{proof}
For proving the result, we need to define some invariant to our problem. For stating this invariant we need to consider additional annotations to our problem. We now consider instances of the form \[(G,k,M,\mH,A_0,A_1,A_2,X,C,G').\] Informally $A_0$ will denote the vertices in $M$ from the beginning, $A_0=M\cap M_0$, $A_1$ the vertices of $M$ added by \ruleref{rl:big}, $A_2$ for the ones added to $M$ by the rule~\ruleref{rl:KrrTilde}, $C$ the set of vertices obtained by contracting at least one edge by \ruleref{rl:deg2}, $X$ the vertices removed by application of \ruleref{rl:neighborhood} (so $X$ is a set of vertices outside of $V(G)$), and $G'$ is a supergraph of $G$ with $V(G')=(V(G)\cup X)$ that will be useful to ``remember'' the links between the vertices of $V(G)$ and $X$. More formally, initially for $(G,M,k,\emptyset)\in \I$ the instance is such that $A_0=M$, $A_1=A_2=X=C=\emptyset$ and $G'=G$. The rules update the additional annotations in the following manner:
\begin{itemize}
    \item When removing vertices from $V(G-M)$ in \ruleref{rl:deg1} and \ruleref{rl:same_neighborhood_prime}, remove the vertices from $G'$ as well and from the sets containing them.
    \item If we apply \ruleref{rl:deg2}, denoting $uv$ the contracted edge in $G$ and $w$ the vertex added to $G$ for replacing $u$ and $v$, we do the same operation in $G'$, with an edge between a vertex $q\in V(G)\cup X$ and $w$ if and only if there was an edge between $q$ and $u$, $q$ and $v$, or both. Moreover we add $w$ to the set $C$.
    \item If we apply \ruleref{rl:big}, the vertex $v$ added to $M$ is added to $A_1$ as well.
    \item If we apply \ruleref{rl:neighborhood}, the vertex $v\in M$ removed from $G$ is now added to $X$, and removed from the set $A_0, A_1$, or $A_2$ it belongs to. The graph $G'$ is not modified.
    \item If we apply \ruleref{rl:KrrTilde}, in the $t$ branches where we add vertices to $M$, we add them to $A_2$ as well now. For the $t$ branches where $t-1$ vertices are removed from $M$, again they are now added to $X$, and removed the sets $A_0, A_1$, and $A_2$.
\end{itemize}
We now state some invariants on those sets:
\begin{claim}\mbox{}
\begin{enumerate}
    \item $M$ is partitioned in $A_0,A_1,A_2$.
    \item $G'[(V(G)\cup X)\setminus C]$ is a $K_{r,r}$-free graph in $\G$.
    \item For $v\in V(G)$, if $v\in C$ or $N_{G'}(v)\cap C\neq \emptyset$ then $d_G(v)\leq 2$.
    \item $C\cap A_1=\emptyset$.
    \item For $v\in A_1$, $|(N_{G'}(v)\cap (M \cup X)) \setminus C|\geq 2\d(r)$.
\end{enumerate}
\end{claim}
\begin{proof}
The first claim is trivial to verify.
For the second, observe that the only modification to this graph is the deletion of vertices, so the result follows by hereditary of $\G$. For each call the degree of a vertex does not increase in $G$ at any time (but it can in $G'$), so it suffices to prove the result when the vertex was just added. If a vertex $v\in V(G)$ is added to $C$ or one of its neighbors it implies that $d_G(v)\leq 2$.
Finally, for the last claim, this property is true when a vertex is added to $A_2$ and there is no contraction of edges or deletion of vertices in $G'[M\cup X]$, so the inequality is kept. The fourth result can be proved by observing that an element of the intersection would have first been in $C$ and then $A_1$. Adding an element in $A_1$ requires that it has $2\d(r)>2$ neighbors in $G$, but this would contradict the previous item. Finally the last inequality is obtained by observing that the considered set stays the same once $v$ is added to $A_1$. And when $v$ is added to $A_1$ we have $d_M(v)\geq 2\d(r)>2$, so $N_{G'}(v)\cap C=\emptyset$ by the third item and so $N_{G'}(v)\subseteq (M\cup X)\setminus C$.\cqed
\end{proof}
Now following the same reasoning than in the sketch but this time using the graph $G'[(V(G)\cup X)\setminus C]$, we obtain $|A_0|\leq 2k_0$, $|A_2|\leq k_0p_2(r,t)$ and $2\d(r)|A_1|\leq \d(r)(|A_0|+|A_1|+|A_2|)$ so $|M|\leq 4k_0+2k_0p_2(r,t)$.
\end{proof}

Using \autoref{lm:sizeGMinusM} and \autoref{lm:sizeM}, the following corollary is now immediate.

\begin{corollary}\label{lm:sizeG-end}
Let $(G,k,M,\mH)$ be an instance of \AFVS.
If on this input \autoref{algo:A} reaches line~\ref{A:DP},
then $|V(G)|=\O(k_0p_6(r,t))$ where $p_6(r,t)=p_2(r,t)p_3(r,t)p_4(r,t)$.
\end{corollary}

\section{Complexity analysis}\label{sec:analysis}

\subsection{Complexity of the dynamic programming in the base case of \autoref{algo:A}}
Let us first describe how \autoref{algo:A} solves an instance  $(G, k, M, \mH)$ using dynamic programming when no rule applies (line~\ref{A:DP}). Recall that by definition of \AFVS, $G$ is a $K_{r,r}$-free graph from a nice graph class $\G$, and $\mH$ a family of disjoint subsets of $M$, with each $H\in \mH$ inducing a connected graph in $G$. 

\begin{theorem}\label{th:solveAFVS}
    An \AFVS instance $(G, k, M, \mH)$ with $G$ an $n$-vertex graph can be solved in time $\tw(G)^{\O(\tw (G))}n^{\O(1)}$.
\end{theorem}
\begin{proof}
For an instance $(G, M, k, \mH)$ with $G$ a graph on $n$ vertices, computing a tree decomposition with bags of size $\O(\tw(G))$ can be done in time $2^{\O(\tw(G))}n$~\cite{korhonen2022single}. Given such a tree decomposition, observe that a bag of size $b$ can intersect at most $b$ distinct $H\in \mH$. Moreover, for $H\in \mH$, the fact that $H$ is connected ensures that the set of bags of the tree decomposition of $G$ containing at least one vertex of $H$ is a connected subset. Using those two observations it is straightforward to adapt the standard dynamic programming algorithm solving \FVS (see \cite{Cygan2015Book}) for solving the more general \AFVS without worsening the running time, and obtaining the claimed complexity.
\end{proof}

Let us call \emph{DP} the algorithm of \autoref{th:solveAFVS}.

\begin{corollary}\label{cor:DP}
Let $(G,k,M,\mH)$ be an instance of \AFVS.
If on this instance \autoref{algo:A} reaches line~\ref{A:DP} and calls
$DP(G,k,M,\mH)$, then the worst case running time of $DP(G,k,M,\mH)$ is  $\tdp(k_0,r,t)=2^{\O\left(k_0^{\delta}p_7(r,t)\right)}$ with $p_7(r,t)=\f(r)\log (\f(r)k_0p_6(r,t))(p_6(r,t))^{\delta}$.
\end{corollary}
\begin{proof}
Denoting $n=|V(G)|$, according to \autoref{lm:sizeG-end}, we have $n=\O(k_0p_6(r,t))$, and as $\G$ is a nice class (see \autoref{def:NCtrees}), we have $\tw(G)\leq \f(r)n^{\delta}$. Thus, \autoref{th:solveAFVS} implies that $DP(G,k,M,\mH)$ runs in time $\left(\f(r)n\right)^{\O\left(\f(r)n^{\delta}\right)}n^{\O(1)}$.
\end{proof}

\subsection{Complexity of \autoref{algo:A}}
Informally, the complexity of \autoref{algo:A} is dominated by the only branching rule~\ruleref{rl:KrrTilde}. The appropriate parameter $\alpha$ associated to an instance $(G,k,M,\mH)$ is $\alpha=k+(k-|\mH|)$, as informally $\mH$ is a packing of hyperedges that we have to hit, and thus $|\mH|$ is a lower bound of the cost of any solution. Thus, if $p(n)$ is the (polynomial) complexity of all operations performed in one call of \autoref{algo:A}, and $f(n,\alpha)$ is the worst complexity of \autoref{algo:A} when $|V(G)|=n$ and $\alpha=k+(k-|\mH|)$, \ruleref{rl:KrrTilde} gives $f(n,\alpha) \le p(n)+(2t)f(n,\alpha-(t-1))$, leading to $f(n,\alpha) \le p(n)(2t)^{\frac{2k-h}{t-1}}$.

\begin{lemma}\label{lm:complexityA}
On an input $(G,k,M,\mH)\in \I$ with $G$ a $n$-vertex graph, \autoref{algo:A} runs in time
$$\tdp\left(k_0,r,t\right)(2t)^{\frac{2k}{t-1}}n^{\O(1)}$$ where $\tdp$ is defined in \autoref{cor:DP}.
\end{lemma}
\begin{proof}
Let $f(n,k,m,h)$ be the worst case complexity of 
\autoref{algo:A} on an instance $(G,k,M,\mH)$ when $|V(G)|=n$, $|M|=m$, and $\mH=h$. We set $C_{DP}=\tdp\left(k_0,r,t\right)$.
Let $p(n)$ be the worst case complexity of all operations performed in one call of \autoref{algo:A} (which is bounded by the sum of the complexities of testing and applying all the kernelization and branching rules, and applying  \algopart). Notice that, as all rules and 
\autoref{sssec:partition-algo} are polynomial, $p$ is polynomial.
Let us show by induction on $2(n+k)-m-h$ that 
\[
f(n,k,m,h) \le (p(n)+C_{DP})(2n-m)\left((2t)^{\frac{2k-h}{t-1}}-1\right).
\]
Without loss of generality, we assume that $p$ is non-decreasing.
Let us bound $f$ according to which recursive call we perform. 
In particular, if we apply \ruleref{rl:neighborhood}, 
we delete one vertex $v$ from $M$, and as we remove from the packing $\mH$ the (at most one) hyperedge containing $v$, we decrease $h$ by $x \in \{0,1\}$.
If we apply \ruleref{rl:KrrTilde}, then we make $t$
recursive calls on instances  $((G-X_{\overline{v}}), k-(t-1), M\setminus X_{\overline{v}}, \mH-X_{\overline{v}})$ (re-using notations of \ruleref{rl:KrrTilde}), and thus in these calls we remove $x \in \intv{0}{t-1}$ hyperedges from $\mH$, 
and in the worst case also $t$ recursive calls on instances $\left(G, k, M\cup R_i, \mH\cup \mT_{R_i}\right)$, where $R_i$ is the union of $t-1$ parts of a $\KttTilde$, each part having size at most $p_2(r,t)$ by definition of a $\KttTilde$.
Thus $f(n,k,m,h)$ is at most:
\begin{enumerate}
    \item $p(n)+f(n-x,k,m,h)$, for $x\geq 1$, when \ruleref{rl:deg1}, \ruleref{rl:deg2}, or \ruleref{rl:same_neighborhood_prime} applies;
    \item $p(n)+f(n,k,m+1,h)$ when \ruleref{rl:big} applies;
    \item $p(n)+f(n-1,k-1,m-1,h-x)$ for $x \in \{0,1\}$ when \ruleref{rl:neighborhood} applies;
    \item $p(n)+ t f(n-(t-1),k-(t-1),m-(t-1), h-x)+ t f(n,k,m+(t-1)p_2(r,t),h+(t-1))$ for $0\leq x\leq t-1$ when \ruleref{rl:KrrTilde} applies; and
    \item $p(n)+\CDP$ if we do not apply any rule and call the dynamic programming algorithm DP.
\end{enumerate}

It remains to check by induction that our bound holds in any of these cases.
The first two and last cases are straightforward.
For the third case, we get
\begin{align*}
    f(n,k,m,h) & \le & p(n)+f(n-1,k-1,m-1,h-x) \\
      & \le & p(n)+(p(n-1)+\CDP)(2n-m-1)\left((2t)^{\frac{2k-h-1}{t-1}}-1\right) \\
      & \le & p(n)+(p(n)+\CDP)(2n-m-1)\left((2t)^{\frac{2k-h-1}{t-1}}-1\right) \\
      & \le & (p(n)+\CDP)(2n-m)\left((2t)^{\frac{2k-h}{t-1}}-1\right)
\end{align*}

For the fourth case, let us first bound the terms $z_1=f(n\! -(t\!-1),k\!-(t\!-1),m\!-(t\!-1),h\!-x)$ and $z_2 = f(n,k,m+(t-1)p_2(r,t),h+(t-1))$ by

\begin{align*}
     z_1 & \le   (p(n-(t-1))+\CDP)(2n-m-(t-1))\left((2t)^{\frac{2k-h-2(t-1)+x}{t-1}}-1\right)  \\
    & \le   (p(n)+\CDP)(2n-m)\left((2t)^{\frac{2k-h}{t-1}-1}-1\right)  \\
     z_2 & \le  (p(n)+\CDP)(2n-m)\left((2t)^{\frac{2k-h}{t-1}-1}-1\right)  
\end{align*}
This leads to the following upper bound for the fourth case:
\begin{align*}
    f(n,k,m,h)  & \le  p(n)+t(z_1+z_2) \\
& \le  p(n)+ (p(n)+\CDP)(2n-m)(2t)\left((2t)^{\frac{2k-h}{t-1}-1}-1\right)\\
& \le \left(p(n)+\CDP\right)(2n-m)\left((2t)^{\frac{2k-h}{t-1}}-1\right).
\end{align*}
This ends the induction. The desired result is then easily obtained from the proven bound.
\end{proof}

Let us now bound the running time of the main algorithm (that branches to create the family of instances $\I$ and runs \autoref{algo:A} on each instance).
\begin{main-thm}
For every nice hereditary graph class $\mathcal{G}$ there is a constant $\eta<1$ such that FVS can be solved in $\mathcal{G}$ in time $2^{k^\eta}\cdot  n^{\O(1)}$.

\end{main-thm}
\begin{proof}
Let denote $n_0=|V(G_0)|$.
According to \autoref{lem:step2}, the complexity to generate the set $\I$ is in $2^{\O\left (r\log k_0 \right )}n_0^{\O(1)}(2r)^{\frac{k_0}{r-1}}$ and $|\I|=\O\left((2r)^{\frac{k_0}{r-1}}\right)$.
Solving \AFVS on an instance $(G,k,M,\mH)\in \I$ with $G$ a $n$-vertex graph can be done in time $\tdp(k_0,r,t)(2t)^{\frac{2k_0}{t-1}}n_0^{\O(1)}$,
where $\tdp(k_0,r,t)=2^{\O\left(k_0^{\delta}p_7(r,t)\right)}$ according to \autoref{lm:complexityA}, \autoref{cor:DP} and the observation that $k\leq k_0$ and $n\leq n_0$. 

So the overall running time is bounded by:
$$\left(2^{\O\left(r\log k_0+\frac{k_0\log r}{r}\right)}+2^{\O\left(\frac{k_0\log r}{r}+\frac{k_0\log t}{t}+k_0^{\delta}p_7(r,t) \right)} \right)n^{\O(1)}.$$
Now lets recall that $t$ is defined as a function of $r$ as $t=2\d(r)=r^{\O(1)}$ and so $p_7(r,t)=r^{\O(1)}$. 
More precisely if we have 
$f_1(r)=\TO\left(r^{c_{f_1}}\right),~f_2(r,p,m)=\TO(r^{c_{f_2}}(p+m)^{c_{f_2}'})$, $\f(r)=\TO(r^{c_f})$ and $\d(r)=\TO(r^{c_d})$ (remember that without loss of generality we supposed $c_d\geq 1$), then we have $p_7(r,t)=\TO(r^{c_7})$ with
\begin{align*}
    c_7&=c_f+\delta(2(c_d+c_{f_2}+c_{f_2'}(c_{f_1}+(6+\alpha)c_d))+c_{f_1}+(6+\alpha)c_d)\\
    &=c_f+\delta(2(c_d+c_{f_2})+(c_{f_2}'+1)(c_{f_1}+(6+\alpha))).
\end{align*}
So by taking $r=k_0^{\eps}$ with $\eps=\frac{1-\delta}{c_7+1}$ we obtain a running time  $2^{\TO\left(k_0^{1-\eps}\right)}n_0^{\O(1)}$. 
To get rid of the hidden logarithmic factors in the exponent and obtain the wanted result it then suffices to replace $\eps$ with some $0<\eps'<\eps$.
\end{proof}

\section{Applications}\label{sec:applications}
In this section, we prove that $s$-string graphs and pseudo-disk graphs are nice graph classes for some parameters. We recall that a system of pseudo-disks is a collection of regions in the plane homeomorphic to a disk such that any two of them share at most 2 points of their boundary.
Similar arguments are used for both considered classes, but in the case of pseudo-disks the arguments are a bit simpler, we then consider this class first.

 In order to give bounds on tree neighborhood complexity we will use the following theorem.
 
 \begin{theorem}[\cite{pseudodisjoint}]\label{hyperpseudo}
    Given $\EuScript{F}$ a family of pseudo-disks and $\EuScript B$ a finite family of pseudo-disks, consider the hypergraph $H(\EuScript B,\EuScript F)$ whose vertices are the pseudo-disks in $\EuScript B$ and the edges are all subsets of~$\EuScript B$ of the form $\{D \in \EuScript B,~D \cap S \neq \emptyset\}$, with $S\in \EuScript F$.
  Then the number of edges of cardinality at most $k\geq 1$ in $H(\EuScript B,\EuScript F)$ is $\O(|\EuScript B|k^3)$.
\end{theorem}
A very important aspect in this result is that it is not required that $\EuScript{F} \cup \EuScript{B}$ is a family of pseudo-disks, and thus pseudo-disks of $\EuScript B$ may ``cross'' pseudo disks of $\EuScript F$.
This is indeed the case in our applications where in particular we associate to each tree in $G-M$ a pseudo-disk in $\EuScript{B}$, and this pseudo disk may cross pseudo-disks associated to vertices of  $\EuScript{M}$ (see proof of \autoref{lm:sstringtreeneigh}).

\subsection{Application to pseudo-disk graphs}
In this section, we prove that the class of pseudo-disk graphs is nice, and so by our main theorem it admits a robust subexponential FPT algorithm for \FVS. Note that for this graph class, the existence of such algorithm was very recently given in~\cite{FVS-WG}.
\begin{lemma}\label{lm:pseudotreeneigh}
    There exists a constant $c$ such that the class of pseudo-disk graphs has bounded tree neighborhood complexity with parameters $\alpha=4,~f_1(r)=c$ and $f_2(r,p,m)=cp^3$.
\end{lemma}
\begin{proof}
Let $r\geq 2$, $G$ be a $K_{r,r}$-free pseudo-disk graph, $(\D_v)_{v\in V(G)}$ be a representation of $G$ as pseudo-disks, $A\subseteq V(G)$, and $\mT$ a family of disjoint non-adjacent trees of $G-A$. We would like to apply \autoref{hyperpseudo} on the set $A$ and the family of subsets of the plane obtained by taking for each tree $T\in \mT$ the union $\D_T=\cup_{v\in T} \D_v$. Because we are considering pseudo-disks and trees, this union is homeomorphic to a disk. Moreover as the trees are non-adjacent the obtained unions do not intersect each other, $(\D_T)_{T\in \mT}$ is trivially a system of pseudo-disks. Now applying \autoref{hyperpseudo} with the families $(\D_v)_{v\in A}$ and $(\D_T)_{T\in \mT}$ directly gives that there exists a constant $c_2$ such that if for all $T\in \mT$ we have $d_A(T)\leq p$ then $|\{N_A(T),~T\in \mT\}|\leq c_2p^3|A|$ as wanted for the second bound. Observing that we can take $p=|A|$ as an upper bound for $d_A(T)$ with $T\in \mT$ gives the first wanted bound.
\end{proof}

\begin{lemma}
    There exist constants $c_1,c_2,c_3,c_4$ such that the class of pseudo-disk graphs is a nice class with parameters $\alpha=4,~f_1(r)=c_1,~f_2(r,p,m)=c_2p^3$, $\delta=\frac 12$, $\f(r)=c_3\sqrt{r\log r}$ and $\d(r)=c_4r\log r$.
\end{lemma}
\begin{proof}
Most of the conditions were already stated with \autoref{th:tw}, \autoref{th:Lee} and \autoref{lm:pseudotreeneigh}. The only property that remains to check is that in a pseudo-disk graph $G$, the contraction of an edge between two degree-two vertices $u$ and $v$ that do not belong to a triangle results in a pseudo-disk graph. 
By considering $(\D_q)_{q\in V(G)}$ a pseudo-disk representation of $G$, we can take $\D_w=\D_u\cup \D_v$ as it is homeomorphic to a disk and its boundary will cross the neighbor of $u$ (respectively $v$) as many times as does the boundary of $\D_u$ (respectively $\D_v$).
\end{proof}

And so applying our main theorem we obtain:
\maincortwo*
Note that our main theorem gives that it suffices to take $\eta>\frac{44}{45}$ for the result to hold. Sharper results already exist in the literature for pseudo-disk graphs, by generalizing the robust algorithm of \cite{lokSODA22} dealing with disk graphs to pseudo-disk graphs, or by using the representation as pseudo-disks for obtaining an even better running time. See \cite{preprintFVSpseudo}, the full version of \cite{FVS-WG}, for both methods.
\subsection{Application to \texorpdfstring{$s$}{s}-string graphs}\label{ssec:sstring}
We now show how to adapt the arguments from the previous section to the case of $s$-strings.

\begin{lemma}\label{lm:sstringtreeneigh}
 There exist constants $c_1,c_2$ such that the class of $s$-string graphs has bounded tree neighborhood complexity with parameter $\alpha=4, f_1(r)=c_1s^4r\log r$ and $f_2(r,p,m)=c_2(sr\log r)^4(p+m)^3$.
\end{lemma}
\begin{proof}
    The proof is very similar to the proof of \autoref{lm:pseudotreeneigh} with some additional tweaks.
    Let $G$ be a $K_{r,r}$-free $s$-string graph for some $r\geq 2$ and $(\mathcal S_v)_{v\in V(G)}$ be an $s$-string representation of $G$.
For $T\in \mT$, again we consider the region $\D_T = \bigcup_{v\in V(T)}\mS_v$.
In this case however, even though $T$ is a tree, the complement of this region is not necessarily connected (for example if two strings of $T$ intersect on several points).
In order to obtain a region of the plane homeomorphic to a disk, we modify $\D_T$ by cutting small sections  where no intersections take place (not even with a string of $A$) and obtain a new set $\D_T'\subseteq \mathcal D_T$ that is connected and such that there is exactly one arc-connected region in $\mathbb{R}^2\setminus \mathcal D_T'$ and which intersects the same strings in $A$ as $\mathcal D_T$ (see \autoref{fig:unitodisk}).
Moreover we add a small thickness to $\D_T'$ while preserving the wanted property about the number of regions in $\mathbb{R}^2\setminus \mathcal D_T'$ in order to obtain a geometric shape homeomorphic to a disk.
Observe that as there are no edges between vertices of distinct members of $\mT$, by taking the thickness small enough the connected sets of the form $\mathcal D_T'$ (for $T\in \mT$) are disjoint and so they trivially form a system of pseudo-disks.

\begin{figure}[ht]
    \centering
    \includegraphics[scale=.8]{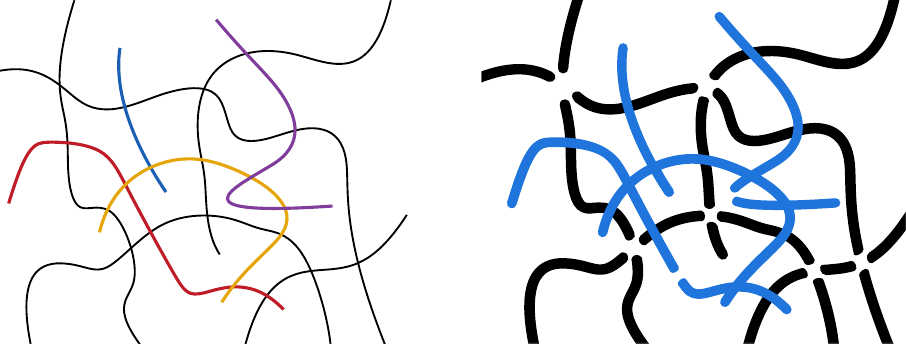}
    \caption{Transformation of a union of $s$-strings to a region homeomorphic to a disk, while keeping the same neighborhood in $A$. The strings of $A$ are depicted in black.}
    \label{fig:unitodisk}
\end{figure}

Now we would like to construct another system of pseudo-disks representing the vertices of $A$. We do this by cutting the strings in $(\mathcal S_u)_{u\in A}$ whenever two of them cross each other. Observe that by \autoref{th:Lee} there exists a constant $c$ such that by denoting $\d(r)=cr\log r$ we have at most $\d(r) |A|$ intersecting pairs, and remember that each pair intersects each other at most $s$ times as we consider a $s$-string representation. So we have at most $s\d(r) |A|$ cuts which gives an upper bound of $(s\d(r)+1)|A|$ for the number of sections of strings between intersection points. We denote $C$ the set of such sections. Again we define a family $(\mathcal R_v)_{v\in C}$ by considering the sections (which are disjoint by construction) and giving them some thickness, while preventing the sections to intersect each other. This is clearly a system of pseudo-disks as its elements are not crossing. By \autoref{hyperpseudo}, the number of distinct sets  $N_C(\mathcal D'_T) = \{v\in C,~{\mathcal D'_T} \text{ intersects } {\mathcal{R}_v}\}$, for $T\in \mT$,
is $\O(x^3|C|)$ where $x$ is the maximum, over every  $T\in \mT$, of the number of pseudo-disks in $(\mathcal R_v)_{v\in C}$ crossed by $D'_T$. This upper bound gives the second result of this lemma by observing that if a pair of pseudo-disks $\D'_{T_1},\D'_{T_2}$ intersect exactly the same set of $\mR_v$, then $N_A(T_1)=N_A(T_2)$. We obtain the first result of the lemma by using the trivial bound $x\leq |C|$.
\end{proof}

\begin{lemma}\label{lm:sstringnice}
    There exist constants $c_1,c_2,c_3,c_4$ such that for every $s\geq 1$ the class of $s$-string graphs is a nice class with parameters $\alpha=4,~f_1(r)=c_1s^4r\log r,~f_2(r,p,m)=c_2(sr\log r)^4(p+m)^3$, $\delta=\frac 12$, $\f(r)=c_3\sqrt{r\log r}$ and $\d(r)=c_4r\log r$.
\end{lemma}
\begin{proof}
Again most of the conditions were already stated with \autoref{th:tw}, \autoref{th:Lee} and \autoref{lm:sstringtreeneigh}. The only remaining property to prove is that given an $s$-string graph $G$, contracting an edge between degree-two vertices $u$ and $v$ that do not belong to a triangle preserves being a $s$-string graph. We do this by replacing in a $s$-string representation $(\mS_q)_{q\in V(G)}$ of $G$ the strings $\mS_u$ and $\mS_v$ with a new string $\mS_w$ while preserving the global property of being an $s$-string family.  $\mS_u$ contains a Jordan arc $a_u$ linking its two neighbors, without being intersected by any string in its interior. Let denote $p_v$ the extremity of $a_u$ in $\mS_v$. Then $\mS_v$ contains a Jordan arc $a_v$ linking $p_v$ to the string representing its second neighbor, without being intersected by any other string than $\mS_u$. The string $\mS_w$ is then obtained by taking the union of $a_u$ and $a_v$. This construction ensures that the number of crossings between two strings other than $\mS_u$ and $\mS_v$ of the initial representation remain the same. We now have to bound the number of crossing of $\mS_w$ with others strings.
Recalling that $u$ and $v$ do not share a common neighbor, and that by construction $w$ has $2$ crossings (one for each distinct neighbor), $w$ does not have more than one crossing with some other string. So after replacing $\mS_u$ and $\mS_v$ by $\mS_w$ we still have a $s$-string family.
\end{proof}
\begin{figure}
    \centering
    \includegraphics[scale=.6]{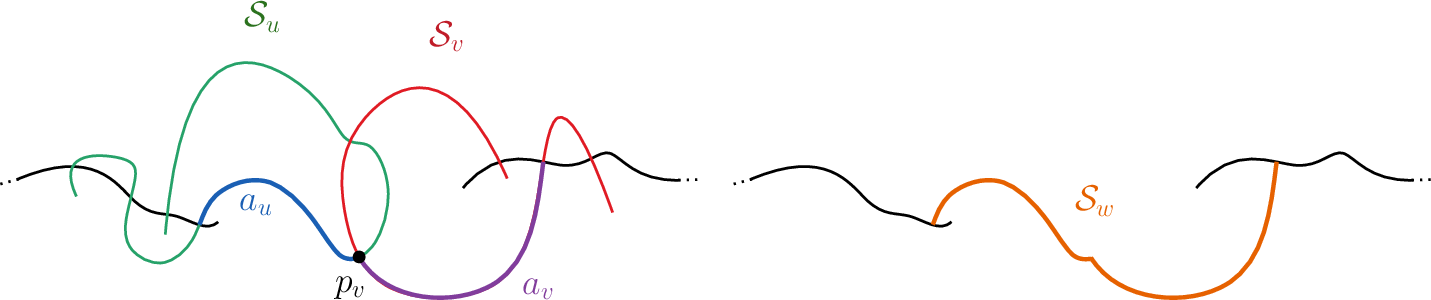}
    \caption{Illustration of the construction used in the proof of \autoref{lm:sstringnice} for contracting an edge between adjacent degree-two vertices without a common neighbor.}
    \label{fig:rule2}
\end{figure}

And so applying our main theorem we obtain:
\maincor*
More precisely, we can take $\eta=\frac{52}{53}$.

\section{Conclusion}\label{sec:ccl}

In this paper we gave general sufficient conditions for the existence of subexponential parameterized algorithms for the \fvs problem. Our main theorem unifies the previously known results on several classes of graphs such as planar graphs, map graphs, unit-disk graphs, disk graphs, or more generally pseudo-disk graphs, and string graphs of bounded degree.
It also applies to classes where such algorithms were not known to exist, notably intersection graphs of thin objects such as segment graphs and more generally $s$-string graphs.

However, we have so far no evidence that our concept of nice class could be an answer to \autoref{q:main}. So a natural direction for future research would be to investigate more general graph classes than those listed above or to discover new classes that fit in our framework.
There are two natural candidates:
\begin{itemize}
    \item Intersection graphs of objects in higher dimensions.  However, as proved in \cite{fomin2018excluded}, under ETH intersection graphs of unit balls in $\RR^3$ do not admit a subexponential parameterized algorithm for \FVS.
    \item Graph classes excluding an induced minor.\footnote{A graph $H$ is said to be an \emph{induced minor} of a graph $G$ if it can be obtained from $G$ by deleting vertices and contracting edges. Otherwise, $G$ is said to \emph{exclude} $H$ as induced minor or to be \emph{$H$-induced minor free}.} This is more general than string graphs (which exclude a subdivided $K_5$ as an induced minor). In such classes, Korhonen and Lokshtanov \cite{korhonen2023induced} recently provided an algorithm solving \FVS in time $2^{\O_H(n^{2/3}\log n)}$, that is, subexponential in the input size~$n$. So far a subexponential FPT algorithm (parameterized by $k$) is still missing.
    \end{itemize}
    
    Less general than the previous item is the class of string graphs, which is nevertheless the most general class of intersection graphs of (arc-connected) objects in the Euclidean plane.  This class is not a nice class, as we are missing the property of bounded tree neighborhood complexity. In fact, even the more restricted case of the class of outerstring graphs\footnote{Outerstring graphs are the string graphs with a representation such that there is a circle containing at least one endpoint of each string.} does not have bounded (tree) neighborhood complexity. 
    This is witnessed by a construction we describe now, that is illustrated in~\autoref{fig:counterNeigh}. For $r>3$, construct an outerstring graph $G$ as follows: consider a clique of size $r$ obtained as an outerstring graph by taking $r$ strings with one extremity on a circle and then permuting the order of the strings in order to obtain the $r!$ permutations. It is then not hard to see that it is possible to add $2^{r}$ disjoint strings with one endpoint on the circle all crossing distinct subsets of the $r$ initial strings. The graph $G$ is then the intersection graph of those $r+2^r$ strings. Observe that $G$ has only $r+1$ vertices with degree at least $r$, and so does not have $K_{r,r}$ as a subgraph. By taking $A\subseteq V(G)$ the vertices associated to the $r$ first strings, we have $|\{N(v)\cap A: v\in V(G)\}|=2^r$. A bound of the form $\O(r^c\cdot |A|)$ for some $c$ is thus impossible.
    Note that we can create examples of arbitrarily large size by taking disjoint copies of this construction. 
    
    \begin{figure}[ht]
    \centering
    \includegraphics[width=0.35\textwidth]{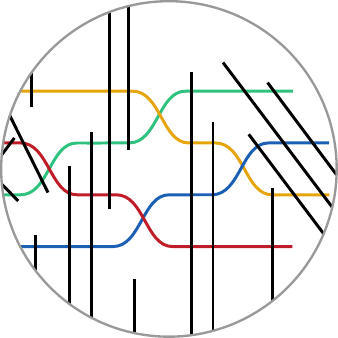}
    \caption{The class of outerstring graphs does not have bounded tree neighborhood complexity, as in a set $A$ of $t$ outerstrings crossing each other in a way their top to bottom order goes through many permutations (the strings in color in the figure, with $t=4$), it is possible to generate $2^t$ distinct neighborhoods in $A$ (in black).}
    \label{fig:counterNeigh}
\end{figure}

    As \autoref{lem:findkrr} allows reducing to $K_{t,t}$-free subgraphs, one way to deal with string graphs would be proving that $K_{t,t}$-free string graphs are $t^c$-string graphs for some constant $c$. Unfortunately, this fails as there are string graphs with $n$ vertices, hence $K_{n,n}$-free, that require $2^{cn}$ crossings for some constant~$c$~\cite{stringexpo}.

\bigskip
\bibliography{biblio} \bibliographystyle{alpha}

\end{document}